\documentclass[12pt, oneside]{article} 
\usepackage[margin=0.95in]{geometry}
\usepackage{graphicx}
\usepackage{amssymb}
\usepackage{amsmath}
\usepackage{setspace}
\usepackage{amsthm}
\usepackage{color}
\usepackage[svgnames]{xcolor}
\usepackage{booktabs}
\usepackage{array} 
\usepackage{caption}
\usepackage[authoryear,round,longnamesfirst,numbers]{natbib}
\usepackage{diagbox}
\usepackage{enumitem}
\usepackage[]{todonotes}
\usepackage{comment}

\usepackage{makecell}
\newcolumntype{C}{>{\centering\arraybackslash $}p{1.3cm}<{$}}
\newcolumntype{D}{>{\centering\arraybackslash $}p{1.4cm}<{$}}
\newcolumntype{S}{>{\centering\arraybackslash $}p{0.8cm}<{$}}

\usepackage{subcaption}
\usepackage{hyperref}
\usepackage[capitalize,nameinlink]{cleveref}
\hypersetup{
colorlinks,breaklinks,naturalnames,
citecolor=DarkBlue,urlcolor=DarkBlue,linkcolor=DarkBlue
}

\usepackage{pgf}
\usepackage{tikz}
\usetikzlibrary{decorations.markings}
\usetikzlibrary{patterns,hobby}
\usetikzlibrary{shapes.misc, positioning}

\newtheorem{lemma}{Lemma}
\newtheorem{proposition}{Proposition}

\newtheorem{theorem}{Theorem}
\newtheorem{definition}{Definition}
\newtheorem{example}{Example}

\newtheorem{observation}{Observation}

\newtheorem*{thmot}{Kantorovich Duality}

\DeclareMathOperator*{\argmax}{arg\,max}

\newcommand{\term}[1]{\textit{#1}}

\makeatletter
\DeclareRobustCommand{\iscircle}{\mathord{\mathpalette\is@circle\relax}}
\newcommand\is@circle[2]{%
  \begingroup
  \sbox\z@{\raisebox{\depth}{$\m@th#1\bigcirc$}}%
  \sbox\tw@{$#1\square$}%
  \resizebox{!}{\ht\tw@}{\usebox{\z@}}%
  \endgroup
}
\makeatother

\renewcommand{\Re}{\mathbb{R}}
\newcommand{\supp}{\text{supp}}
\newcommand{\proj}{\text{proj}_M}

\renewcommand{\epsilon}{\varepsilon}

\newcommand*\circled[1]{\tikz[baseline=(char.base)]{
            \node[shape=circle,draw,inner sep=1.6pt] (char) {#1};}}

\definecolor{Burgundy}{RGB}{206,5,40}
\definecolor{Rust}{RGB}{200,10,10}
\definecolor{UCSDblue}{RGB}{55,80,200}

\title{Credible Persuasion\footnote{We are indebted to Nageeb Ali for his continuing guidance and support. We are also grateful for comments and suggestions from Ian Ball, Carl Davidson, Henrique de Oliveira, Piotr Dworczak, Alex Frankel,  Nima Haghpanah, Marc Henry, Tetsuya Hoshino, Yuhta Ishii, Emir Kamenica, Navin Kartik, Vijay Krishna,  Meg Meyer, Moritz Meyer-ter-Vehn,  Harry Pei, Daniel Rappoport, Andrew Rhodes, Ron Siegel, Alex Smolin, Juuso Toikka, Jia Xiang, Takuro Yamashita, Hanzhe Zhang, and participants at various conferences and seminars.}}

\begin{document}

\author{Xiao Lin\thanks{Department of Economics, Pennsylvania State University (e-mail: \href{mailto: xiao@psu.edu}{xiao@psu.edu}).} \and Ce Liu\thanks{Department of Economics, Michigan State University (e-mail: \href{mailto: celiu@msu.edu}{celiu@msu.edu}).}}

\date{May 6, 2022}

\maketitle

\begin{abstract}
We propose a new notion of credibility for Bayesian persuasion problems. A disclosure policy is credible if the sender cannot profit from tampering with her messages while keeping the message distribution unchanged. We show that the credibility of a disclosure policy is equivalent to a cyclical monotonicity condition on its induced distribution over states and actions. We also characterize how credibility restricts the Sender's ability to persuade under different payoff structures. In particular, when the sender's payoff is state-independent, all disclosure policies are credible. We apply our results to the market for lemons, and show that no useful information can be credibly disclosed by the seller, even though a seller who can commit to her disclosure policy would perfectly reveal her private information to maximize profit.

\end{abstract}

\thispagestyle{empty}

\clearpage
\setcounter{tocdepth}{2}
\tableofcontents
\thispagestyle{empty}
\clearpage

\setcounter{page}{1}

\section{Introduction} \label{section: introduction}

When an informed party (Sender; she) discloses information to persuade her audience (Receiver; he), it is in her interest to convey only messages that steer the outcome in her own favor: schools may want to inflate their grading policies to improve their job placement records; similarly, credit rating agencies may publish higher ratings in exchange for future business. Even when the Sender claims to have adopted a disclosure policy, she may still find it difficult to commit to following its prescriptions, since the adherence to such policies is often impossible to monitor. By contrast, what \textit{is} often publicly observable is the final distribution of the Sender's messages: students' grade distributions at many universities are publicly available, and so is the distribution of rating scores from credit rating agencies.

Motivated by this observation, we propose a notion of \textit{credible persuasion}. In contrast to standard Bayesian persuasion, our Sender cannot commit to a disclosure policy; however, to avoid detection, she must keep the final message distribution unchanged when deviating from her disclosure policy. For example, in the context of grade distributions, if a university had announced a disclosure policy that features a certain fraction of A's, B's, and C's, it cannot switch to a distribution that assigns every student an A without being detected. Analogously, if a credit rating agency were to tamper with its rating schemes, any resulting change in the overall distribution of ratings may be detected. Our notion of credibility closely adheres to this definition of detectability: we say that a disclosure policy is {credible} if given how the Receiver reacts to her messages, the Sender has no profitable deviation to any other disclosure policy that has the same message distribution.

Can the Sender persuade the Receiver by using credible disclosure policies? We find that in many settings, no informative disclosure policy is credible. An important special case where this effect is exhibited is the market for lemons \citep{Akerlof1970}. Here, we show that the seller of an asset cannot credibly disclose any useful information to the buyer; this effect arises even though the seller benefits from persuasion when she can fully commit to her disclosure policy. Conversely, we also provide conditions for when the Sender is guaranteed to benefit from credible persuasion so that credibility does not entirely eliminate the scope for persuasion. In general, we show that credibility is characterized by a \textit{cyclical monotonicity} condition that is analogous to that studied in decision theory and mechanism design \citep{rochet1987necessary}.

\medskip
To illustrate these ideas, consider the following example. A buyer (Receiver) chooses whether to buy a car from a used car seller (Sender). It is common knowledge that $30\%$ of the cars are of \textit{high} quality and the remaining $70\%$ are of \textit{low} quality. For simplicity, suppose that all cars are sold at an exogenously fixed price.\footnote{In \cref{section: applications} we study a competitive market for lemons with endogenous prices, and emerge with similar findings.} The payoffs in this example are in \cref{example: adverse selection payoffs}.
The seller always prefers selling a car, but the buyer is only willing to purchase if and only if he believes its quality is {high} with at least $0.5$ probability. Conditional on a car being sold, the seller obtains the same payoff regardless of its quality; but when a car is not sold, she receives a higher value from retaining a {high} quality car.

\begin{table}[t]
\centering
\begin{subtable}[h]{{0.4\textwidth}}
\centering
{\setlength{\extrarowheight}{2pt}
\begin{tabular}{C | D  D  }
\Xhline{3\arrayrulewidth}
  & \text{\small Buy} & \text{\small Not Buy} \\
\Xhline{3\arrayrulewidth}
\text{\small High}  & 2  & 1  \\
\text{\small Low} & 2 & 0 \\
\Xhline{3\arrayrulewidth}
\end{tabular}
}
\caption*{Seller}
\end{subtable}
\hspace{2ex}
\begin{subtable}[h]{0.4\textwidth}
\centering
{\setlength{\extrarowheight}{2pt}
\begin{tabular}{C | D  D}
\Xhline{3\arrayrulewidth}
   & \text{\small Buy} & \text{\small Not Buy} \\
\Xhline{3\arrayrulewidth}
\text{\small High}  & 1  & 0  \\
\text{\small Low} & -1 & 0 \\
\Xhline{3\arrayrulewidth}
\end{tabular}
}
\caption*{Buyer}
\end{subtable}
\caption{Used Car Example Payoffs\label{example: adverse selection payoffs}}
\end{table}

As a benchmark, let us first see what the seller achieves if she could commit to a disclosure policy. We depict the optimal disclosure policy in \Cref{figure: used car full-commitment disclosure}. The policy uses two messages, \textit{pass} and \textit{fail}: all high-quality cars pass, along with $3/7$ of the low-quality cars; the remaining $4/7$ of the low-quality cars receive a failing grade. Conditional on the car passing, the buyer believes that the car is of high quality with probability $0.5$, which is just enough to convince him to make the purchase. If a car fails, the buyer believes that the car is of low quality for sure and will refuse to buy.  With this disclosure policy, the buyer expects to see the seller pass $60\%$ of the cars and fail the remaining $40\%$.

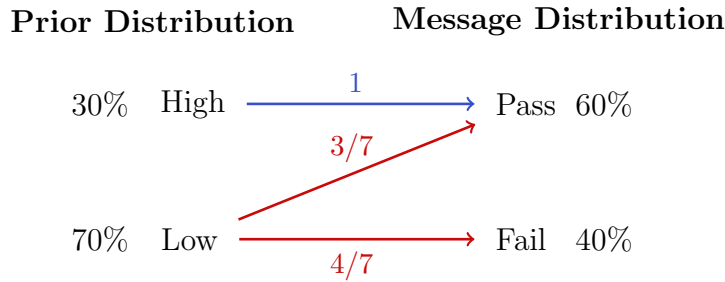
\begin{figure}[ht]
\centering
\begin{tikzpicture}[shorten > = 4pt, shorten <= 4pt]
{
\path (0.6,1.8) node (state1)[left,xshift=6]{Low};
\path (0.7,3.6) node (state2)[left,xshift=6]{High};

\path (-0.2,4.7) node (text1){\textbf{Prior Distribution}};
\path (5.2,4.7) node (text1){\textbf{Message Distribution}};

\path (-0.9, 1.8) node (prior1) {$70\%$};
\path (-0.9, 3.6) node (prior2) {$30\%$};

\path (4.5,1.8) node (m1)[right,xshift=-8]{Fail};
\path (4.5,3.6) node (m2)[right,xshift=-8]{Pass};

\draw[] (5.8,1.8) node (posterior1) {$40\%$};
\draw[] (5.8,3.6) node (posterior2) {$60\%$};
}

\tikzset{every path/.append style={->, line width=1pt}}
\draw[UCSDblue] (state2) -- (m2) node[pos=0.48,above] {\small 1};
\draw[Rust] (state1) -- (m2) node[above, pos=0.48] {\small 3/7};
\draw[Rust] (state1) -- (m1) node[below, pos=0.48] {\small 4/7};
\end{tikzpicture}
\caption{Optimal Commitment Policy \label{figure: used car full-commitment disclosure}}
\end{figure}

The policy above is optimal for the seller if she can commit to following its prescriptions. But suppose the buyer cannot observe how the seller rates her cars. Instead, the buyer only observes the fraction of cars being passed and failed.  In such a setting, the seller can profitably deviate from the above disclosure policy without being detected by the buyer.  Specifically, the seller can switch to failing all high-quality cars and passing an equal number of low-quality cars. This disclosure policy, illustrated in \Cref{figure: undetectable deviation}, induces the same distribution of messages (i.e., $60\%$ pass, $40\%$ fail). Holding fixed the buyer's behavior, this deviation is profitable for the seller because she is still selling the same number of cars but now is able to retain more high-quality cars. Accordingly, we view the optimal full-commitment policy not to be credible: after having promised to share information according to a disclosure policy, the seller would not find it rational to follow through and would instead profit from an undetectable deviation.

\begin{figure}[ht]
\centering

\begin{tikzpicture}[shorten > = 4pt, shorten <= 4pt]
{
\path (0.6,1.8) node (state1)[left,xshift=6]{Low};
\path (0.7,3.6) node (state2)[left,xshift=6]{High};

\path (-0.2,4.7) node (text1){\textbf{Prior Distribution}};
\path (5.2,4.7) node (text1){\textbf{Message Distribution}};

\path (-0.9, 1.8) node (prior1) {$70\%$};
\path (-0.9, 3.6) node (prior2) {$30\%$};

\path (4.5,1.8) node (m1)[right,xshift=-8]{Fail};
\path (4.5,3.6) node (m2)[right,xshift=-8]{Pass};

\draw[] (5.8,1.8) node (posterior1) {$40\%$};
\draw[] (5.8,3.6) node (posterior2) {$60\%$};
}
\tikzset{every path/.append style={->, line width=1pt}}

\draw[Rust] (state1) -- (m2) node[below, pos=0.38] {\footnotesize 6/7};
\draw[Rust] (state1) -- (m1) node[pos=0.38, below] {\footnotesize 1/7};
\draw[UCSDblue] (state2) -- (m1) node[pos=0.38, above] {\footnotesize 1};

\end{tikzpicture}
\caption{An Undetectable Deviation \label{figure: undetectable deviation}}
\end{figure}
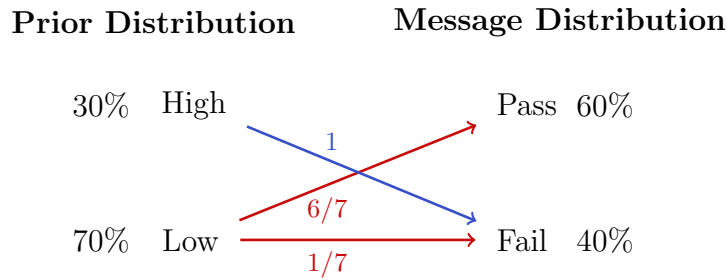

More generally, we introduce the following notion of credibility for disclosure policies. 
Consider a \emph{profile} consisting of the Sender's disclosure policy and the Receiver's strategy (mapping messages to actions). We say that a profile is \emph{Receiver incentive compatible} if the Receiver's strategy best responds to the Sender's disclosure policy---this requirement is standard in Bayesian persuasion problems. We say that a profile is \emph{credible} if, given the Receiver's strategy, the Sender has no profitable deviation to any other disclosure policy that induces the same message distribution. Together, credibility and Receiver incentive compatibility require that conditional on the Sender's message distribution, the Sender and Receiver best respond to each other.\footnote{Our solution-concept is therefore analogous to a Nash equilibrium condition in which the set of feasible deviations for the Sender is to other disclosure policies that induce the same message distribution.}

We have just argued that in the used car example, the optimal full-commitment disclosure policy was not credible given the Receiver's best response. Can any car be sold in a profile that is both credible and Receiver incentive compatible? We show that the answer is no. Note that this is what happens when no information is disclosed. In other words, credibility completely shuts down the possibility for useful information transmission.

To see why, suppose towards a contradiction that the buyer purchases a car after observing a message $m_1$ that is sent with positive probability. By Receiver incentive compatibility, the buyer must believe that the car is of high quality with more than $0.5$ probability after observing $m_1$. Since $m_1$ is sent with positive probability, the martingale property of beliefs implies that there must be another message $m_2$, also sent with positive probability, that makes the buyer assign less than $0.5$ to the car's quality being high. Necessarily, when the buyer observes the message $m_2$, he does not make a purchase. This creates an incentive for the seller to tamper with her disclosure policy: by exchanging some of the good cars being mapped into $m_1$ with an equal number of bad cars being mapped into $m_2$, she can improve her payoff without  changing the distribution of messages.

One may wonder if credibility always shuts down communication entirely. The next example features a setting in which the optimal full-commitment disclosure policy is credible. Consider the disclosure problem faced by a school (Sender) and an employer (Receiver).\footnote{See \cite{ostrovsky2010information} for an early study of how schools strategically design their grading policies in a competitive setting.} Just as in the used car example, a student's ability is either \textit{high} with probability $0.3$ or \textit{low} with probability $0.7$. Payoffs are as shown in \Cref{example: school payoffs}.  The employer is willing to hire a student if he believes the student has high ability with at least $0.5$ probability. The school would like all its students to be employed, but derives a higher payoff from placing a good student than it does from placing a bad one.

\begin{table}[h]
\centering
\begin{subtable}[h]{{0.4\textwidth}}
\centering
{\setlength{\extrarowheight}{2pt}
\begin{tabular}{C | D  D  }
\Xhline{3\arrayrulewidth}
 & \text{\small Hire} & \text{\small Not Hire} \\
\Xhline{3\arrayrulewidth}
\text{\small High}  & 2  & 0  \\
\text{\small Low } & 1 & 0 \\
\Xhline{3\arrayrulewidth}
\end{tabular}
}
\caption*{School}
\end{subtable}
\hspace{2ex}
\begin{subtable}[h]{0.4\textwidth}
\centering
{\setlength{\extrarowheight}{2pt}
\begin{tabular}{C | D  D}
\Xhline{3\arrayrulewidth}
 & \text{\small Hire} & \text{\small Not Hire} \\
\Xhline{3\arrayrulewidth}
\text{\small High}  & 1  & 0  \\
\text{\small Low} & -1 & 0 \\
\Xhline{3\arrayrulewidth}
\end{tabular}
}
\caption*{Employer}
\end{subtable}
\caption{School Example Payoffs\label{example: school payoffs} }
\end{table}

The school's optimal full-commitment disclosure policy is identical to the one in the used car example (\Cref{figure: used car full-commitment disclosure}), and so are the employer's best responses. But unlike the used car example, the school cannot profitably deviate without changing the message distribution.

To see why, note that without changing the message distribution, any deviation must involve passing some low ability students while failing an equal number of {high} ability students. This would increase the employment of {low} ability students at the expense of their {high} ability counterparts, which  makes the school worse off. Since the school cannot profit from undetectable deviations, the optimal full-commitment policy is credible. In contrast to the previous example where credibility shuts down all useful communication, the current example shows that credibility sometimes imposes no cost on the Sender relative to persuasion with full commitment.

\medskip

In the two examples above, credibility has starkly different implications for information transmission. The key difference is that in the used car example, when the car's quality is higher, the Sender has a weaker incentive to trade while the Receiver's incentive to trade is stronger; in the school example, by contrast, both the Sender and Receiver have a stronger incentive to trade as the student's ability increases. Our results formalize this intuition.

\cref{proposition: noinformation} shows that when the Sender and Receiver's preferences have opposite modularities (e.g. when the Sender's payoff is strictly supermodular and the Receiver's payoff is submodular), no useful information can be credibly communicated. When players' preferences share the same modularity, the Sender does not always benefit from credible persuasion relative to the no-information benchmark. \cref{proposition: alignedpreference} and \cref{proposition: no cost} provide additional conditions that guarantee the Sender does benefit from credible persuasion, as well as conditions under which the optimal full-commitment disclosure policy is credible.

Generalizing further, we use optimal transport theory to characterize credibility using a familiar condition from mechanism design and decision theory---cyclical monotonicity. \cref{theorem-cyclical} shows that for every profile of Sender's disclosure policy and Receiver's strategy, the credibility of the profile is equivalent to a {cyclical monotonicity} condition on its induced  distribution over states and actions. As is illustrated in the examples above, credibility requires that the Sender cannot benefit from any pairwise swapping in the matching of states and actions. The cyclical monotonicity condition generalizes this idea to cyclical swapping: for every sequence of state-action pairs in the support, the sum of the Sender's utility should be lower after the matchings of states and actions in this sequence are permuted. In \cref{section: relationship}, we discuss the connection of \cref{theorem-cyclical} to \cite{rochet1987necessary}.

We apply our results to a classic setting of asymmetric information---the market for lemons. It is well-known that market outcome in such a setting may be inefficient due to adverse selection \citep*{Akerlof1970}: despite common knowledge of gain from trade, some cars may not be traded. If the seller can commit to a disclosure policy to persuade the buyers, she can completely solve the market inefficiency by perfectly revealing $\theta$ to the buyers. However, we show that if the buyers can only observe the message distribution of the seller's disclosure policy, but not exactly how these messages are generated, then the seller cannot credibly disclose any useful information to the buyer.

Our paper also offers foundations for studying Bayesian persuasion in a number of settings. One example is when the Sender's payoff is state-independent: in these cases, our results imply that all disclosure policies are credible, so the full-commitment assumption in the Bayesian persuasion approach is nonessential as long as the message distribution is observable. Another example is when the Sender's payoff is supermodular, in which case all monotone disclosure policies are credible.

\medskip
The rest of the paper is organized as follows: \cref{section: model} introduces our credibility notion as well as the main results. \cref{section: applications} considers an application: in the market for lemons with endogenous prices, we show that the seller cannot credibly disclose any useful information to the buyers, even though full disclosure would maximize the seller's profit. \cref{section: discussions} discusses several aspects of our model: first, it connects our cyclical-monotonicity characterization to \cite{rochet1987necessary}; it then discusses a finite-population approximation of our model, and provides an extensive-form foundation to our partial-commitment formulation. \cref{section: conclusion} concludes. All omitted proofs are in \cref{section: appendix}. The remainder of this introduction places our contribution within the context of the broader literature.

\paragraph{Related Literature:}

Our work contributes to the study of strategic communication. The Bayesian persuasion model in \citet{kamenicagentzkow2011} studies a Sender who can fully commit to an information structure.\footnote{\citet{brocas2007influence} and \cite{rayosegal2010} also study optimal disclosure policy in more specific settings.} By contrast, the cheap-talk approach pioneered by \citet{crawford1982strategic} models a Sender who observes the state privately and, given the Receiver's strategy, chooses an optimal (sequentially rational) message. The partial commitment setting that we model is between these two extremes: here, the Sender can commit to a (marginal) distribution over messages but not the entire information structure.

Our model considers a Sender who can misrepresent her messages as long as the misrepresentation still produces the original message distribution. This contrasts with existing approaches to modeling limited commitment in Bayesian persuasion. One approach, pioneered by \citet*{lipnowski2019persuasion} and \citet{min2021bayesian}, is to allow the Sender to alter the messages from her chosen test with some fixed probability. Another approach is to consider settings where the Sender can revise her test at a cost. \citet{nguyen2021} consider a Sender who can distort the messages from her chosen information structure, whereas \citet{perezskreta} consider a Sender who can falsify the state, or input, of the information structure.

The way that we model the Sender's feasible deviations is closely related to the literature on quota mechanisms, which use message budgets to induce truth-telling;
see, for example, \cite{jackson2007overcoming}, \cite{matsushima2010role}, \cite{rahman2010detecting}, and \cite{frankel2014aligned}. Similar ideas have also been explored in communication games. For example, \cite{chakraborty2007comparative} consider multi-issue cheap-talk problems, and study equilibria where the Sender assigns a ranking to each issue.\footnote{In such equilibria, the Sender's on-path deviations are essentially restricted to permuting the assignment of rankings to issues. This is similar to quota mechanisms where deviations are restricted to permuting the matching between the Sender's true types and her reported types.} \cite{renault2013dynamic}  study repeated cheap-talk models where only messages and the Receiver's actions are publicly observable. They characterize equilibria in the repeated communication game via a static reporting game where the Sender directly reports her type. The key condition in their characterization requires truthful reporting to be optimal among all reporting strategies that reproduce the true type distribution, which is akin to \cite{rahman2010detecting}'s characterization of implementable direct mechanisms.\footnote{Relatedly, \cite{escobar2013efficiency} consider repeated private-value Bayesian games with Markovian states, and show that efficiency can be achieved using a similar quota-like construction.} \cite{margaria2018dynamic} use a different approach to study the case where the Sender's payoff is state-independent, and \cite{meng2021value} provides a unified approach to characterizing the Receiver's optimal value in these repeated cheap-talk models. \citet*{kuvalekar2021goodwill} study a related model where the Receiver is short-lived, and show that the equilibrium payoffs can be characterized via a static cheap-talk model with capped money burning.

A different strand of the repeated cheap-talk literature studies models where the Receiver can observe feedback about past state realizations. \citet{best2020persuasion} considers how coarse feedback of past states can substitute for commitment; \citet*{mathevet2022reputation} allows for the possibility of non-strategic commitment types; \cite{pei2020repeated} characterizes when Sender's persistent private information about her own lying cost allows her to achieve her full-commitment payoff.

Finally, our approach to credible persuasion is reminiscent of how \citet{akbarpour2020credible} model credible auctions. They study mechanism design problems where the designer's deviations are ``safe'' so long as they lead to outcomes that are possible when she is acting honestly, and characterize mechanisms that ensure the designer has no safe and profitable deviations. By contrast, we study persuasion problems where the Sender's deviations are undetectable if they do not alter the message distribution, and characterize information structures where the Sender has no profitable and undetectable deviation.

\section{Model} \label{section: model}
\subsection{Setup}\label{section: modelsetup}

We consider an environment with a single Sender ($S$; she) and a single Receiver ($R$; he). Both players' payoffs depend on an unknown state $\theta\in \Theta$ and the Receiver's action $a\in A$. Both $\Theta$ and $A$ are finite sets. The payoff functions are given by $u_S:\Theta\times A\rightarrow \mathbb{R}$ and $u_R:\Theta\times A \rightarrow \mathbb{R}$. Players hold full-support common prior $\mu_{0}\in\Delta(\Theta)$.

Let $M$ be a finite message space that contains $A$. The Sender chooses a disclosure policy, which we henceforth refer to as a ``test," to influence the Receiver's action. A test $\lambda \in \Delta(\Theta\times M)$ is a joint distribution of states and messages, so that the marginal distribution of states agrees with the prior; that is, $\lambda_\Theta=\mu_{0}$.\footnote{For a probability measure $P$ defined on some product space $X\times Y$, we use $P_{X}$ and $P_Y$ to denote its marginal distribution on $X$ and $Y$, respectively.} The Receiver chooses an action after observing each message according to a pure strategy $\sigma:M\rightarrow A$.\footnote{We focus on pure strategies to abstract from the Receiver using randomization to deter the Sender's deviations.}

Our interest is in understanding the Sender's incentives to deviate from her test, which depends on the Receiver's strategy. To avoid ambiguity, we refer explicitly to pairs of $(\lambda, \sigma)$---or \textit{profiles}---that consist of a Sender's disclosure policy and a Receiver's strategy. For each profile $(\lambda, \sigma)$, the players' expected payoffs are 
\begin{equation*}
U_S(\lambda,\sigma)= \sum_{\theta,m} u_S(\theta,\sigma(m))\lambda(\theta,m)\qquad\text{and}\qquad 
U_R(\lambda,\sigma)= \sum_{\theta,m} u_R(\theta,\sigma(m))\lambda(\theta,m).
\end{equation*}

We consider a setting where the Sender cannot commit to her test, and can deviate to another test so long as it leaves the final message distribution unchanged. This embodies the notion that the distribution of the Sender's messages is observable, even though it may be difficult to observe exactly how these messages are generated.
Formally, if $\lambda$ is a test promised by the Sender, let $D(\lambda)\equiv \{\lambda'\in \Delta(\Theta \times M): \lambda'_\Theta=\mu_0,\lambda'_M=\lambda_M\}$
denote the set of tests that induce the same distribution of messages as $\lambda$: these test are indistinguishable from $\lambda$ from the Receiver's perspective. Our credibility notion requires that conditioning on how the Receiver responds to the Sender's messages, no deviation in $D(\lambda)$ can be profitable for the Sender.

\begin{definition}\label{definition: credible}
A profile  $(\lambda,\sigma)$ is \textbf{credible} if
\begin{equation}\label{IC-S}
\lambda \in \argmax_{\lambda'\in D(\lambda)} \; \sum_{\theta,m}\; u_S(\theta,\sigma(m))\;\lambda'(\theta,m)
\end{equation}
\end{definition}

Moreover, the Receiver's strategy is required to be a best response to the Sender's chosen test.

\begin{definition} \label{definition: RIC}
A profile $(\lambda,\sigma)$ is \textbf{Receiver Incentive Compatible} (R-IC) if
\begin{equation}\label{equation: RIC}
	\sigma\in \argmax_{\sigma': M \rightarrow A} \; \sum_{\theta, m}\; u_R(\theta,\sigma'(m))\;\lambda(\theta,m)
\end{equation}
\end{definition}
Together, credibility and R-IC ensure that conditioning on the message distribution of the Sender's test, both the Sender and the Receiver best respond to each other.

An immediate observation is that there always exists a profile $(\lambda,\sigma)$ that is both credible and R-IC. This is the profile of a completely uninformative test and a Receiver strategy that takes the ex ante optimal action after every message. Given the test, the Receiver is taking a best response and given the Receiver's strategy, the Sender has no incentive to deviate to any other test that induces the same message distribution.

\bigskip

Our credibility notion is motivated by the observability of the Sender's message distribution, which is modeled as a restriction on the Sender's feasible deviations. Below we first discuss the observability of message distribution in light of the population interpretation of Bayesian persuasion models; we then connect this observability to the partial commitment approach we adopt in our model.

The observability of the message distribution is best understood through a population interpretation of persuasion models,\footnote{For a more detailed discussion of various interpretations of Bayesian persuasion models, see e.g. Section 2.2 of \citet{kamenica2019bayesian}.} where there are a continuum of objects with types distributed according to $\mu_0\in\Delta(\Theta)$. The Sender's test $\lambda$ assigns each object a message based on its type, which generates a message distribution $\lambda_M$. Working with a continuum population affords us a cleaner exposition by abstracting from sampling variation.  In \cref{section: finite sample}, we consider a finite approximation where the Sender privately observes $N$ i.i.d. samples from $\mu_0\in \Delta(\Theta)$, and assigns each realization a message $m\in M$ subject to quotas on message frequencies; the Receiver then chooses an action after observing the Sender's message. We show that credible and R-IC profiles in our continuum model are approximated by those in the finite-sample model when the sample size $N$ becomes large.

The credibility notion in \autoref{definition: credible} restricts the Sender's feasible deviations by requiring them to preserve the original message distribution. To formalize the connection between the observability of message distribution and the Sender's partial commitment, in \cref{section: extensivesingle} we present an extensive-form game in which the Sender is permitted to deviate to any test, and the Receiver observes the message distribution generated by the chosen test. As we show therein, if the Sender chooses a test that induces a different message distribution, the Receiver can ``punish'' the Sender by assuming that the chosen test is uninformative. If one focuses on profiles where the Sender is weakly better off than disclosing no information (as we do in this paper), this analysis provides a justification for our focus on deviations that generate the same distribution of messages as the equilibrium test.

\subsection{Stable Outcome Distributions}

We characterize credible and receiver incentive compatible profiles through the induced probability distribution of states and actions.
Formally, an \term{outcome distribution} is a distribution $\pi \in \Delta(\Theta\times A) $ that satisfies $\pi_\Theta=\mu_{0}$: this is a consistency requirement that stipulates that the marginal distribution of states must conform to the prior. We say an outcome distribution $\pi$ is induced by a profile $(\lambda,\sigma)$ if for every $(\theta,a)\in \Theta \times A$, $\pi(\theta,a)=\lambda(\theta,\sigma^{-1}(a))$, where $\sigma^{-1}$ is the inverse mapping of $\sigma$.
We are interested in characterizing outcome distributions that can be induced by profiles that are both credible and R-IC, and refer to such outcome distributions as stable.
\begin{definition}\label{Definition-Stable}
An outcome distribution $\pi\in\Delta(\Theta\times A)$ is \textbf{stable} if it is induced by a profile $(\lambda,\sigma)$ that is both credible and R-IC.
\end{definition}

Our first result characterizes stable outcome distributions. 
\begin{theorem}\label{theorem-cyclical}
An outcome distribution $\pi\in\Delta (\Theta\times A)$ is stable if and only if:
\begin{enumerate}
 	\item $\pi$ is $u_R-$obedient: for each $a\in A$ such that $\pi(\Theta, a)>0$,
	\begin{equation*}
 	\sum_{\theta\in \Theta} \pi(\theta,a)\; u_R(\theta,a)\geq \sum_{\theta\in  \Theta}\pi(\theta,a)\; u_R(\theta,a') \;\text{ for all } \;a'\in A.
 	 \end{equation*}
 	\item $\pi$ is $u_S-$cyclically monotone: for each sequence $(\theta_1,a_1),\ldots,(\theta_n,a_n)\in \supp(\pi)$ and $a_{n+1}\equiv a_1$,
	\begin{equation*}
	\sum_{i=1}^n u_S(\theta_i,a_i)\geq \sum_{i=1}^n u_S(\theta_i,a_{i+1});
	\end{equation*}
\end{enumerate}
\end{theorem}
The first condition is the standard obedience constraint \citep{BM2016,taneva2019information}, which specifies that the Receiver finds it incentive compatible to follow the recommended action given the belief that she forms when receiving that recommendation. The second condition, namely $u_S$-cyclical monotonicity, is the new constraint that maps directly to our notion of credibility. Below, we describe this condition and explain why it is both necessary and sufficient for credibility.

To understand the cyclical monotonicity condition, consider an outcome distribution $\pi$ and a sequence $(\theta_i,a_i)_{i=1}^n$ in the support of $\pi$. A ``cyclical'' deviation in this case consists of subtracting $\varepsilon$ mass from $(\theta_i,a_i)$ while adding it to $
(\theta_i,a_{i+1})$ for each $i=1,\ldots,n$, where $a_{n+1}\equiv a_1$. Each step of this cyclical deviation changes the Sender's payoff by
$\varepsilon \big[u_S(\theta_i,a_{i+1} ) -u_S(\theta_i,a_{i} )\big]$,
so the total change in the Sender's payoff is
\begin{equation*}
\varepsilon \big[\sum_{i=1}^n u_S(\theta_i,a_{i+1} )  - \sum_{i=1}^n u_S(\theta_i,a_{i} )\big].
\end{equation*}
The cyclical monotonicity condition requires that the Sender can find no profitable cyclical deviations.

To see why cyclical monotonicity is necessary, observe that cyclical deviations do not change the distribution of recommended actions. Therefore, any such deviation could not be detected solely on the basis of the distribution of messages. Because we require that such undetectable deviations are not profitable, this implies the cyclical monotonicity condition above.

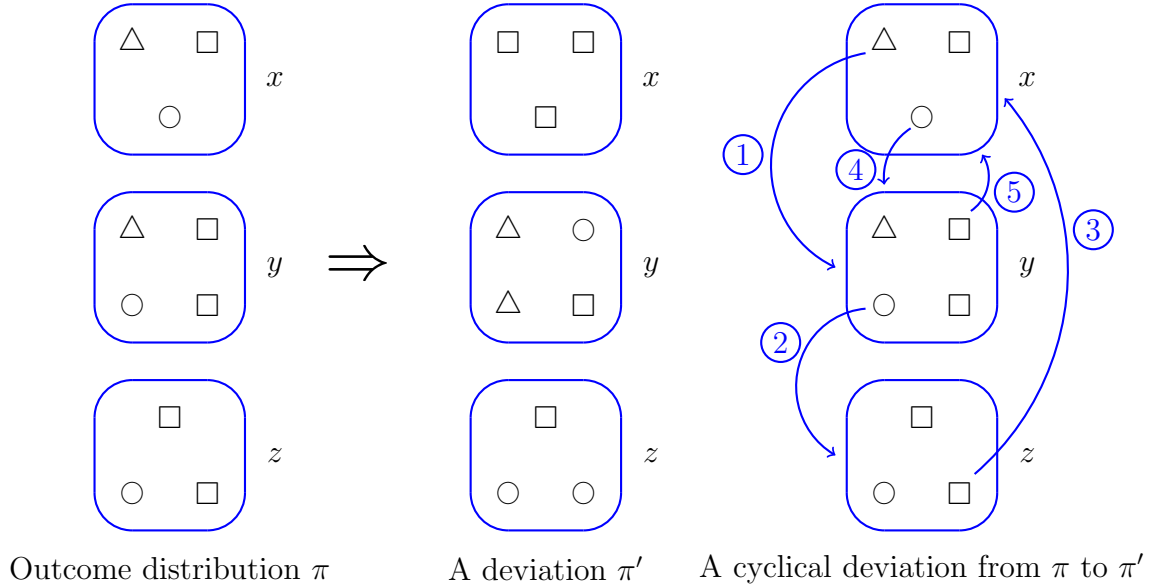
\begin{figure}[ht]
\centering
	\begin{tikzpicture}[domain=0:3, scale=5, thick]    

	\def\x{0.2};
	\def\y{0.2};
	\def\s{0.25};
	\filldraw (0,0)node{$\bigtriangleup$};
	\filldraw (0+\x,0)node{$\Box$} ;
	\filldraw (0.1,0-\y)node{$\iscircle$};

	\draw[blue]  (0,-\y-0.1)--(0+\x,-\y-0.1) (0+\x+0.1,-\y)--(0+\x+0.1,0) (0+\x,0+0.1)--(0,0+0.1) (0-0.1,0)--(0-0.1,-\y);

	\draw[blue] (0-0.1,-\y) arc[start angle=-180, end angle=-90, x radius=0.1, y radius=0.1];
	\draw[blue] (0+\x+0.1,-\y) arc[start angle=0, end angle=-90, x radius=0.1, y radius=0.1];	
	\draw[blue] (0+\x+0.1,0) arc[start angle=0, end angle=90, x radius=0.1, y radius=0.1];
	\draw[blue] (0-0.1,0) arc[start angle=180, end angle=90, x radius=0.1, y radius=0.1];

	\filldraw (0,-0.5)node{$\bigtriangleup$}   ;
	\filldraw (0+\x,-0.5)node{$\Box$} (0+\x,0-\y-0.5)node{$\Box$};
	\filldraw[] (0,0-\y-0.5)node{$\iscircle$}   ;	
	
	\draw[blue]  (0,-\y-0.1-0.5)--(0+\x,-\y-0.1-0.5) (0+\x+0.1,-\y-0.5)--(0+\x+0.1,0-0.5) (0+\x,0+0.1-0.5)--(0,0+0.1-0.5) (0-0.1,0-0.5)--(0-0.1,-\y-0.5);

	\draw[blue] (0-0.1,-\y-0.5) arc[start angle=-180, end angle=-90, x radius=0.1, y radius=0.1];
	\draw[blue] (0+\x+0.1,-\y-0.5) arc[start angle=0, end angle=-90, x radius=0.1, y radius=0.1];	
	\draw[blue] (0+\x+0.1,0-0.5) arc[start angle=0, end angle=90, x radius=0.1, y radius=0.1];
	\draw[blue] (0-0.1,0-0.5) arc[start angle=180, end angle=90, x radius=0.1, y radius=0.1];

	\filldraw (0.1,0-1)node{$\Box$} (0+\x,0-\y-1)node{$\Box$};
	\filldraw[] (0,0-\y-1)node{$\iscircle$}   ;	
	
	\draw[blue]  (0,-\y-0.1-1)--(0+\x,-\y-0.1-1) (0+\x+0.1,-\y-1)--(0+\x+0.1,0-1) (0+\x,0+0.1-1)--(0,0+0.1-1) (0-0.1,0-1)--(0-0.1,-\y-1);

	\draw[blue] (0-0.1,-\y-1) arc[start angle=-180, end angle=-90, x radius=0.1, y radius=0.1];
	\draw[blue] (0+\x+0.1,-\y-1) arc[start angle=0, end angle=-90, x radius=0.1, y radius=0.1];	
	\draw[blue] (0+\x+0.1,0-1) arc[start angle=0, end angle=90, x radius=0.1, y radius=0.1];
	\draw[blue] (0-0.1,0-1) arc[start angle=180, end angle=90, x radius=0.1, y radius=0.1];

    \draw (0.6,-0.6)node{\huge $\Rightarrow$};
    
	\draw[] (0.1,-1.4)node{Outcome distribution $\pi$};
	\draw[] (1.1,-1.4)node{A deviation $\pi'$};
	\draw[] (2.1,-1.4)node{A cyclical deviation from $\pi$ to $\pi'$};

    \draw[] (0.38,0.15-0.25)node[]{$x$};
    \draw[] (0.38,-0.35-0.25)node[]{$y$};
    \draw[] (0.38,-0.85-0.25)node[]{$z$};
    
    \draw[] (1.38,0.15-0.25)node[]{$x$};
    \draw[] (1.38,-0.35-0.25)node[]{$y$};
    \draw[] (1.38,-0.85-0.25)node[]{$z$};
    
    \draw[] (2.38,0.15-0.25)node[]{$x$};
    \draw[] (2.38,-0.35-0.25)node[]{$y$};
    \draw[] (2.38,-0.85-0.25)node[]{$z$};
	\def\z{1};
	\filldraw (0+\z,0)node{$\Box$};
	\filldraw (0+\x+\z,0)node{$\Box$} ;
	\filldraw (0.1+\z,0-\y)node{$\Box$};

	\draw[blue]  (0+\z,-\y-0.1)--(0+\x+\z,-\y-0.1) (0+\x+0.1+\z,-\y)--(0+\z+\x+0.1,0) (0+\x+\z,0+0.1)--(0+\z,0+0.1) (0-0.1+\z,0)--(0-0.1+\z,-\y);

	\draw[blue] (0-0.1+\z,-\y) arc[start angle=-180, end angle=-90, x radius=0.1, y radius=0.1];
	\draw[blue] (0+\x+0.1+\z,-\y) arc[start angle=0, end angle=-90, x radius=0.1, y radius=0.1];	
	\draw[blue] (0+\x+0.1+\z,0) arc[start angle=0, end angle=90, x radius=0.1, y radius=0.1];
	\draw[blue] (0-0.1+\z,0) arc[start angle=180, end angle=90, x radius=0.1, y radius=0.1];

	\filldraw (0+\z,-0.5)node{$\bigtriangleup$}   ;
	\filldraw (0+\x+\z,-0.5)node{$\iscircle$} (0+\x+\z,0-\y-0.5)node{$\Box$};
	\filldraw[] (0+\z,0-\y-0.5)node{$\bigtriangleup$}   ;	
	
	\draw[blue]  (0+\z,-\y-0.1-0.5)--(0+\z+\x,-\y-0.1-0.5) (0+\z+\x+0.1,-\y-0.5)--(0+\z+\x+0.1,0-0.5) (0+\z+\x,0+0.1-0.5)--(0+\z,0+0.1-0.5) (0-0.1+\z,0-0.5)--(0+\z-0.1,-\y-0.5);

	\draw[blue] (0-0.1+\z,-\y-0.5) arc[start angle=-180, end angle=-90, x radius=0.1, y radius=0.1];
	\draw[blue] (0+\x+0.1+\z,-\y-0.5) arc[start angle=0, end angle=-90, x radius=0.1, y radius=0.1];	
	\draw[blue] (0+\x+0.1+\z,0-0.5) arc[start angle=0, end angle=90, x radius=0.1, y radius=0.1];
	\draw[blue] (0-0.1+\z,0-0.5) arc[start angle=180, end angle=90, x radius=0.1, y radius=0.1];
	
	\filldraw (0.1+\z,0-1)node{$\Box$} (0+\x+\z,0-\y-1)node{$\iscircle$};
	\filldraw[] (0+\z,0-\y-1)node{$\iscircle$}   ;	
	
	\draw[blue]  (0+\z,-\y-0.1-1)--(0+\x+\z,-\y-0.1-1) (0+\x+0.1+\z,-\y-1)--(0+\x+0.1+\z,0-1) (0+\x+\z,0+0.1-1)--(0+\z,0+0.1-1) (0-0.1+\z,0-1)--(0-0.1+\z,-\y-1);

	\draw[blue] (0-0.1+\z,-\y-1) arc[start angle=-180, end angle=-90, x radius=0.1, y radius=0.1];
	\draw[blue] (0+\x+0.1+\z,-\y-1) arc[start angle=0, end angle=-90, x radius=0.1, y radius=0.1];	
	\draw[blue] (0+\x+0.1+\z,0-1) arc[start angle=0, end angle=90, x radius=0.1, y radius=0.1];
	\draw[blue] (0-0.1+\z,0-1) arc[start angle=180, end angle=90, x radius=0.1, y radius=0.1];

	\def\w{2};
	\filldraw (0+\w,0)node{$\bigtriangleup$};
	\filldraw (0+\x+\w,0)node{$\Box$} ;
	\filldraw (0.1+\w,0-\y)node{$\iscircle$};

	\draw[blue]  (0+\w,-\y-0.1)--(0+\w+\x,-\y-0.1) (0+\x+0.1+\w,-\y)--(0+\x+0.1+\w,0) (0+\x+\w,0+0.1)--(0+\w,0+0.1) (0-0.1+\w,0)--(0-0.1+\w,-\y);

	\draw[blue] (0-0.1+\w,-\y) arc[start angle=-180, end angle=-90, x radius=0.1, y radius=0.1];
	\draw[blue] (0+\x+0.1+\w,-\y) arc[start angle=0, end angle=-90, x radius=0.1, y radius=0.1];	
	\draw[blue] (0+\x+0.1+\w,0) arc[start angle=0, end angle=90, x radius=0.1, y radius=0.1];
	\draw[blue] (0-0.1+\w,0) arc[start angle=180, end angle=90, x radius=0.1, y radius=0.1];

	\filldraw (0+\w,-0.5)node{$\bigtriangleup$}   ;
	\filldraw (0+\x+\w,-0.5)node{$\Box$} (0+\x+\w,0-\y-0.5)node{$\Box$};
	\filldraw[] (0+\w,0-\y-0.5)node{$\iscircle$}   ;	
	
	\draw[blue]  (0+\w,-\y-0.1-0.5)--(0+\x+\w,-\y-0.1-0.5) (0+\x+0.1+\w,-\y-0.5)--(0+\x+0.1+\w,0-0.5) (0+\x+\w,0+0.1-0.5)--(0+\w,0+0.1-0.5) (0-0.1+\w,0-0.5)--(0-0.1+\w,-\y-0.5);

	\draw[blue] (0-0.1+\w,-\y-0.5) arc[start angle=-180, end angle=-90, x radius=0.1, y radius=0.1];
	\draw[blue] (0+\x+0.1+\w,-\y-0.5) arc[start angle=0, end angle=-90, x radius=0.1, y radius=0.1];	
	\draw[blue] (0+\x+0.1+\w,0-0.5) arc[start angle=0, end angle=90, x radius=0.1, y radius=0.1];
	\draw[blue] (0-0.1+\w,0-0.5) arc[start angle=180, end angle=90, x radius=0.1, y radius=0.1];

	\filldraw (0.1+\w,0-1)node{$\Box$} (0+\x+\w,0-\y-1)node{$\Box$};
	\filldraw[] (0+\w,0-\y-1)node{$\iscircle$}   ;	
	
	\draw[blue]  (0+\w,-\y-0.1-1)--(0+\x+\w,-\y-0.1-1) (0+\x+0.1+\w,-\y-1)--(0+\x+0.1+\w,0-1) (0+\x+\w,0+0.1-1)--(0+\w,0+0.1-1) (0-0.1+\w,0-1)--(0-0.1+\w,-\y-1);

	\draw[blue] (0-0.1+\w,-\y-1) arc[start angle=-180, end angle=-90, x radius=0.1, y radius=0.1];
	\draw[blue] (0+\x+0.1+\w,-\y-1) arc[start angle=0, end angle=-90, x radius=0.1, y radius=0.1];	
	\draw[blue] (0+\x+0.1+\w,0-1) arc[start angle=0, end angle=90, x radius=0.1, y radius=0.1];
	\draw[blue] (0-0.1+\w,0-1) arc[start angle=180, end angle=90, x radius=0.1, y radius=0.1];
	
	\draw[->,blue] (-0.05+\w,-0.03) to [curve through={(-0.3+\w,-0.3)}](-0.13+\w,-0.6);
	\draw[blue] (-0.3+\w,-0.3)node[left,xshift=-3,shape=circle,draw,inner sep=1.8pt]{1};
	
	\draw[->,blue] (-0.05+\w,-0.71) to [curve through={(-0.2+\w,-0.8)}](-0.13+\w,-1.1);	
	\draw[blue] (-0.2+\w,-0.8)node[left,xshift=-3,shape=circle,draw,inner sep=1.8pt]{2};
	
	\draw[->,blue] (0.24+\w,-1.15) to [curve through={(0.48+\w,-0.5)}](0.32+\w,-0.15);
	\draw[blue] (0.48+\w,-0.5)node[right,xshift=3,shape=circle,draw,inner sep=1.8pt]{3};
	
	\draw[->,blue] (0.07+\w,-0.23) to [curve through={(0+\w,-0.34)}](0+\w,-0.38);
	\draw[blue] (0+\w,-0.34)node[left,xshift=-3,shape=circle,draw,inner sep=1.8pt]{4};
	
	\draw[->,blue] (0.23+\w,-0.45) to [curve through={(0.27+\w,-0.4)}](0.26+\w,-0.3);
	\draw[blue] (0.27+\w,-0.4)node[right,xshift=3,shape=circle,draw,inner sep=1.8pt]{5};

	\end{tikzpicture}   
	\caption{Cyclical Deviation}
	\label{figure:cyclical}
\end{figure}

For sufficiency, the intuition is that any deviation that keeps the marginal distribution on messages unchanged can either be expressed as or approximated by cyclical deviations. To see the idea, let us look at the graphical representation of an outcome distribution $\pi$ in the left panel of \autoref{figure:cyclical}. In this example, $\Theta=\{\square, \iscircle, \triangle\}$ and $A=\{x,y,z\}$. Each $\square$, $\iscircle$, and $\triangle$ in the graph is associated with $10\%$ probability mass. The prior belief assigns $20\%$ to $\triangle$, which is represented by the presence of two $\triangle$'s; similarly, the prior assigns $30\%$ to $\iscircle$ and $50\%$ to $\square$. The pairing between states and actions pins down the outcome distribution, as well as its induced distribution of actions. For example, $\pi(\triangle, x)=10\%$ and $\pi(\Theta,y)=40\%$. 

The middle panel depicts a possible deviation $\pi'$ that maintains the same distribution of actions. In particular, the matchings among states and actions are permuted, but the number of shapes matched to each action remains the same. Our credibility notion requires that no such deviation can be profitable. The right panel illustrates how this deviation can be expressed as a $5$-step cyclical deviation. 

In fact, it is not hard to see from the graph that every deviation that involves moving integer numbers of $\square$, $\iscircle$, or $\triangle$ is essentially a permutation of the locations of the shapes, which can be written as a disjoint composition of cycles. It is therefore sufficient to check for deviations that form a cycle. However, an implicit assumption of the graphical representation is that both the original outcome distribution and the deviation contain only rational probabilities; otherwise, we cannot find an $N$ large enough so that every joint probability can be represented as integer numbers of shapes. In this case, however, one can approximate irrational outcome distributions with rational ones, and then apply the same argument outlined above. In \cref{section: appendix proof of thm 1}, we prove \cref{theorem-cyclical} by directly citing Kantorovich duality from optimal transport theory.

The final remark is that in the cyclical monotonicity condition, we can restrict attention to those sequences with length $n\leq \max\{|\Theta|,|M|\}$. As an illustration, the right panel of \Cref{figure:cyclical} represents a cyclical deviation of length $5$, which is greater than both the number of states and the number of actions. However, we can split this cycle into two shorter cyclical deviations, one comprising $\circled{1}\,\circled{2}\,\circled{3}$ and the other made up of $\circled{4}\,\circled{5}$. If the deviation of length $5$ is profitable, one of these two shorter sequences has to be profitable. \autoref{theorem-cyclical} therefore reduces the problem to checking only a finite number of deviations.

\subsection{The Case of Additively Separable Payoffs}
If $u_S(\theta,a)$ is additively separable in $\theta$ and $a$, then $u_S-$cyclical monotonicity is automatically satisfied. So we have the following observation.
\begin{observation}\label{observation: additive}
If $u_S(\theta,a) = v(\theta) + w(a)$ for some $v:\Theta\rightarrow\Re$ and $w:A\rightarrow\Re$, then every outcome distribution that satisfies $u_R$-obedience is stable.
\end{observation}
Therefore, in this case, there is no gap between what is achievable by a sender who can fully commit to a test relative to a sender who can only partially commit to a distribution of messages. This observation is relevant to a special and widely studied case of additively separable preferences, namely that in which the Sender has state-independent payoffs.
State-independent payoffs feature in many analyses of communication and persuasion \citep*[e.g.][]{chakraborty2010persuasion,alonso2016persuading,lipnowski2020cheap,lipnowski2019persuasion,gitmez2022polarization}. Our analysis suggests that for these settings, the Sender can persuade with full commitment power even without committing to a test, merely by making public (and committing to) her distribution of messages.

\subsection{When is Credibility Restrictive?}\label{section: whenisrestrictive}
When the state and action interact in the Sender's payoff, credibility limits the Sender's choice of tests.  The goal of this section is to understand how these limits can restrict the Sender's ability to persuade the Receiver. As benchmarks, we will often draw comparisons to what the Sender can achieve when she can fully commit to her disclosure policy, as well as what is achievable when all or no information is disclosed. We say an outcome distribution $\pi^*$ is an \textit{optimal full-commitment outcome} if it maximizes the Sender's payoff among outcome distributions that satisfy $u_R$-obedience. 
An outcome distribution $\hat{\pi}$ is a \textit{fully revealing outcome} if the Receiver always chooses a best response to every state; that is,
\begin{equation*}
a\in \argmax_{a'\in A} u_R(\theta, a')\; \text{ for every }\; (\theta,a)\in \supp (\hat{\pi}).
\end{equation*}
Finally, an outcome distribution $\pi^{\circ}$ is a \textit{no-information outcome} if the Receiver always chooses the same action that best responds to the prior belief $\mu_{0}$; in other words, there exists 
\begin{equation*}
a^*\in \argmax_{a\in A}\sum_{\theta\in\Theta}\mu_{0}(\theta)u_R(\theta,a) \; \text{ such that } \;  \pi^{\circ}_A( a^*) =1.
\end{equation*}

We say the Sender \emph{benefits from persuasion} if an optimal full-commitment outcome gives the Sender a higher payoff than every no-information outcome. Similarly, we say the Sender \emph{benefits from credible persuasion} if there exists a stable outcome distribution that gives the Sender a strictly higher payoff than every no-information outcome.

We make a few assumptions for ease of exposition. First, suppose every $a\in A$ is a best response to some belief $\mu\in\Delta(\Theta)$ for the receiver. This assumption is without loss of generality, since an action that is never a best response would never be played by the Receiver in any R-IC profile, and can be removed from the action set $A$ without changing results in this paper. Second, suppose there exists no distinct $a,a'\in A$ such that $u_R(\theta,a)=u_R(\theta,a')$ for all $\theta\in \Theta$; in other words, from the Receiver's perspective, there are no duplicate actions. This second assumption is \emph{not} without loss, but greatly simplifies the statements of \cref{proposition: noinformation} and \cref{proposition: alignedpreference}.

\paragraph{Modular Preferences:}
In the examples in \cref{section: introduction}, we see that whether the Sender can credibly persuade the Receiver depends crucially on the alignment of their marginal incentives to trade. To understand this logic more generally, we assume that $\Theta$ and $A$ are totally ordered sets, which without loss of generality can be assumed to be subsets of $\Re$. 
Recall that a payoff function $u:\Theta\times A\rightarrow \mathbb{R}$ is \emph{supermodular} if for all $\theta>\theta'$ and $a> a'$, we have
\begin{equation*}
u(\theta,a)+u(\theta',a')\geq u(\theta,a')+u(\theta',a).
\end{equation*}
and \textit{submodular} if 
\begin{equation*}
u(\theta,a)+u(\theta',a')\leq  u(\theta,a')+u(\theta',a).
\end{equation*}
Furthermore, the function is strictly supermodular or strictly submodular if the inequalities above are strict for $\theta>\theta'$ and $a> a'$.

The modularity of players' payoff functions captures how the marginal utility from higher actions varies with the state. This generalizes the marginal incentive to trade in the examples in \cref{section: introduction}: intuitively, the Sender and the Receiver have {aligned marginal incentives} when both players' payoff functions share the same modularity, and {opposed marginal incentives} when their payoff functions have opposite modularities. To fix ideas, we will assume that the Sender's  payoff is supermodular and vary the modularity of the Receiver's payoff.

We now introduce a lemma that simplifies the $u_S$-cyclical monotonicity condition in \cref{theorem-cyclical} when the Sender's payoff is supermodular. Say that an outcome distribution $\pi\in\Delta (\Theta\times A)$ is \textit{comonotone} if for all $(\theta,a)$, $(\theta',a')\in \supp(\pi)$ satisfying $\theta<\theta'$, we have $a\leq a'$. Comonotonicity requires that the states and the Receiver's actions are positive-assortatively matched in the outcome distribution.
The following lemma, whose variant appears in \citet{rochet1987necessary}, shows that $u_S$-cyclical monotonicity reduces to comonotonicity when the Sender's preference is supermodular.
\begin{lemma}\label{lemma: comonotone}
If $u_S$ is supermodular, then every comonotone outcome distribution is  $u_S$-cyclically monotone. Furthermore, if $u_S$ is strictly supermodular, then every $u_S$-cyclically monotone outcome distribution is also comonotone.
\end{lemma}

Combined with \cref{theorem-cyclical}, \cref{lemma: comonotone} implies that when the Sender's preference is strictly supermodular, the credibility of a profile $(\lambda, \sigma)$ is equivalent to the comonotonicity of its induced outcome distribution. Comonotone outcome distributions have attracted much attention in the persuasion literature in part due to their simplicity and ease of implementation; for example, see \citet{dworczak2019simple}, \citet{goldstein2018stress}, \citet{mensch2019monotone}, \citet{ivanov2020optimal}, \citet{kolotilin2018optimal}, and \citet{kolotilin2020relational}. Our credibility notion provides an additional motive for focusing on monotone information structures.

\bigskip

\paragraph{When Credibility Shuts Down Communication:}
The next result generalizes the used-car example in \cref{section: introduction}. 
\begin{proposition}\label{proposition: noinformation}
 If $u_S$ is strictly supermodular and $u_R$ is submodular, then every stable outcome distribution is a no-information outcome.
\end{proposition}

\autoref{proposition: noinformation} says that when the players have opposed marginal incentives, credibility considerations completely shut down information transmission. The logic generalizes what we saw in the example: if two distinct messages resulted in different actions from the Receiver, the Sender and Receiver would have diametrically opposed preferences regarding which action to induce in which state. Therefore, if a profile satisfies R-IC and is even partially informative, the Sender would have a profitable deviation to another test that swaps states and induces the same marginal distribution of messages.

\bigskip

\paragraph{When the Sender Benefits from Credible Persuasion:}

In light of the school example in \cref{section: introduction}, one might expect credibility to not limit the Sender's ability to persuade when her marginal incentives are aligned with the Receiver's. However, this is false without imposing additional assumptions. For an illustration, consider the following example, in which  both the Sender and Rceiver have supermodular payoffs. The Sender benefits from persuasion when she can fully commit to her disclosure policy, but no stable outcome distribution can give her a higher payoff than the best no-information outcome.

\begin{figure}[h]  \centering
\begin{minipage}[b]{0.45\textwidth} \centering

\begin{center}
\begin{tabular}{|c|S|S|S|S|}
\hline
    $u_S(\theta,a)$       &  a_1  &  a_2  &  a_3  &  a_4  \\ \hline
$\theta=H$ &  -1   &  0.75    &  1    &  0    \\ \hline
$\theta=L$ &  0    &  0.75    &  0.5  &  -1   \\ \hline
\end{tabular}
\end{center}

\vspace{0.05in}
\begin{center}
\begin{tabular}{|c|S|S|S|S|}
\hline
    $u_R(\theta,a)$       & a_1 & a_2 &  a_3  &  a_4  \\ \hline
$\theta=H$ &  0    &  0.6  &  0.8  &  1    \\ \hline
$\theta=L$ &  1    &  0.8  &  0.6  &  0    \\ \hline
\end{tabular}   
\end{center}
\vspace{-0.2in}
\end{minipage} 
\hspace{3ex}
\begin{minipage}[b]{0.4\textwidth} 
\centering
\begin{tikzpicture}[domain=0:3, scale=5, thick]

\draw[<->] (1.1,0)node[right]{$\mu$}--(0,0)node[left,yshift=-5]{0}--(0,0.9)node[above]{${v}(\mu)$};

\draw (1.02,0)node[below]{1};

\draw[blue] (0,0)--(0.25,-0.25)  (0.25,0.75)--(0.5, 0.75) (0.5,0.75)--(0.75,0.75+0.125) (0.75,-0.25)--(1,0);

\draw[red,thick] (0,0)--(0.25,0.75)--(0.75,0.75+0.125)--(1,0);

\draw[ dashed] (0.6,0)node[below]{$\mu_0=0.6$}--(0.6,0.84);

\draw (0.125,0)node[above]{\small$a_1$} (0.375,0)node[above]{\small$a_2$} (0.625,0)node[above,xshift=2]{\small$a_3$} (0.875,0)node[above]{\small$a_4$};

\draw (0.25,0)--(0.25,0.02) (0.5,0)--(0.5,0.02) (0.75,0)--(0.75,0.02) (1,0)--(1,0.02)	;
\end{tikzpicture}
\vspace{-0.3in}
\end{minipage}
\vspace{0.4in}

\begin{minipage}[t]{0.45\textwidth}
\captionof{table}{Sender and Receiver's payoffs \label{table: example concavify}}    
\end{minipage}
\begin{minipage}[t]{0.4\textwidth}
\caption{Concavification\label{figure: example concavification}}
\end{minipage}
\end{figure}

\begin{example}\label{example: alignednocommunication}
Suppose $\Theta= \{H,L\}$ with prior $\mu_{0} = P(\theta=H)=0.6$ and $A= \{a_1,a_2,a_3,a_4\}$. The Sender and Receiver's payoffs are as given in \Cref{table: example concavify}. Note that both players' payoffs are strictly supermodular. The Receiver's best response is $a_1$ when $\mu_0\in [0,0.25)$, $a_2$ when $\mu_0\in [0.25,0.5)$, $a_3$ when $\mu_0\in [0.5,0.75)$, and $a_4$ when $\mu_0\in [0.75,1]$; this leads to the Sender's indirect utility function (blue) and its concave envelope (red) depicted in \Cref{figure: example concavification}. From \citet{kamenicagentzkow2011}, the red line represents the Sender's optimal value under full commitment. It is clear that at $\mu_0=0.6$, the Sender strictly benefits from persuasion if she can fully commit to her test.

However, no stable outcome distribution can make the Sender better off than the no-information outcome. To see why, consider any stable outcome distribution that has at least two actions in the support. By \autoref{lemma: comonotone}, at most one of these actions can be matched with more than one state, for otherwise the outcome distribution would not be comonotone. So at most one of the actions in the support of the outcome distribution can be induced by interior posterior beliefs. However, it is clear that in order for the Sender to benefit from persuasion, she must induce both $a_2$ and $a_3$, both of which can only happen when the Receiver holds interior posterior beliefs. As a result, no stable outcome distribution can make the Sender better off.

\end{example}

\Cref{example: alignednocommunication} above shows that besides the co-modularity of preferences, additional conditions are needed in order to ensure the Sender can benefit from credible persuasion. \cref{proposition: alignedpreference} below offers several such conditions.

Let $\overline{a} \equiv \max A$ and $\underline{a} \equiv \min A$ denote the highest and lowest Receiver actions, and let $\overline{\theta} \equiv \max\Theta$ and $\underline{\theta} \equiv \min\Theta$ denote the highest and lowest states.

\begin{proposition}\label{proposition: alignedpreference}
Suppose both $u_S$ and $u_R$ are supermodular. 
\begin{enumerate}
    \item If the highest action is dominant for the Sender, that is,  if $u_S(\theta,\overline{a} ) > u_S(\theta,{a} )$ for all $\theta$ and $a\ne \overline{a}$, then for generic priors,\footnote{Formally, by generic we mean a set of priors $T\subset\Delta(\Theta) $ with the same Lebesgue measure as $\Delta(\Theta)$.} the Sender benefits from credible persuasion as long as she benefits from persuasion.
    \item If the Sender favors extreme actions in extreme states, that is, if $u_S(\overline{\theta},\overline{a} ) > u_S(\overline{\theta}, {a} )$ for all $a\ne \overline{a}$ and $u_S(\underline{\theta},\underline{a} ) > u_S(\underline{\theta}, {a} )$ for all $a\ne \underline{a}$, then for generic priors, the Sender benefits from credible persuasion.
    \item If the Sender is strictly better off from a fully revealing outcome than from every no-information outcome, then the Sender benefits from credible persuasion.
\end{enumerate}
\end{proposition}

The first condition in \cref{proposition: alignedpreference} is satisfied in settings like the school example, where the school would always want to place a student regardless of the student's ability. The second condition is applicable in environments where both parties have agreement on extreme states. For example, both doctors and patients favor aggressive treatment if the patient's condition is severe, and both favor no treatment if the patient is healthy, but they might disagree in intermediate cases. Lastly, a special case of the third condition is quadratic loss preferences as commonly used in models of communication \citep[e.g.][]{crawford1982strategic}.

The first two parts of \autoref{proposition: alignedpreference} are based on belief splitting. Let us briefly describe the proof under the first condition; the proof for the second part follows similar arguments. Note that if $\overline{a}$ is a dominant action for the Sender, and the Sender can benefit from persuasion (under full commitment), then $\overline{a}$ must not already be a best response for the Receiver under the prior $\mu_{0}$. The Sender can then split the prior into a point mass posterior $\delta_{\overline{\theta}}$ and some other posterior $\tilde{\mu}$ that is close to $\mu_{0}$. At $\delta_{\overline{\theta}}$, the Receiver is induced to choose $\overline{a}$ since his payoff is supermodular. In addition, for generic priors the Receiver's best response to $\tilde{\mu}$ remains the same as his best response to $\mu_{0}$. The Sender benefits from this belief-splitting since the same action is still played most of the time, but  in addition her favorite action is now played with positive probability. Moreover, the resulting outcome distribution matches higher states with higher actions, so it is stable due to the supermodularity of $u_S$ and \autoref{lemma: comonotone}.

The intuition for the third part of \autoref{proposition: alignedpreference} is straightforward to see when the Sender's payoff is strictly supermodular. 
Consider $(\theta,a)$ and $(\theta', a')$ in the support of a fully revealing outcome distribution $\pi$, so
$a$ and $a'$ best respond to $\theta$ and $\theta'$, respectively. From \citet{topkis2011}, it follows that $a\geq a'$ if $\theta>\theta'$.
Therefore, $\pi$ is comonotone and satisfies $u_S$-cyclical monotonicity by \autoref{lemma: comonotone}. By construction, $\pi$ also satisfies obedience, so $\pi$ is stable by \autoref{theorem-cyclical}.

\paragraph{When Credibility Imposes No Cost to the Sender:} In \cref{observation: additive}, we see that when the Sender's payoff is additively separable, credibility does not restrict the set of stable outcomes. \cref{proposition: no cost} below provides a condition which guarantees that credibility imposes no loss on the Sender's optimal value, even when credibility does restrict the set of stable outcomes. 

\begin{proposition} \label{proposition: no cost}
Suppose $|A|=2$. If both $u_S$ and $u_R$ are supermodular, then at least one optimal full-commitment outcome is stable; if in addition $u_S$ is strictly supermodular, then every optimal full-commitment outcome is stable.
\end{proposition}
\cref{proposition: no cost} says that in setting where both players have supermodular payoffs and the Receiver faces a binary decision, such as ``accept'' or ``reject'', then credibility imposes no cost to the Sender. This result follows from combining our \cref{theorem-cyclical} and \cref{lemma: comonotone} with Theorem 1 in \citet{mensch2019monotone}. He shows that under the assumptions in our \cref{proposition: no cost}, there exists a optimal full-commitment outcome that is comonotone. The intuition is that for any outcome distribution $\pi$ that is $u_R$-obedient but not comonotone, the Sender can weakly improve her payoff by swapping the non-comonotone pairs in the support of $\pi$, so that they become matched assortatively. Such swapping also benefits the Receiver due to the supermodularity of $u_R$, so $u_R$-obedience remains satisfied. As a result, the Sender can always transform a non-comonotone outcome distribution into one that is comonotone without violating $u_R$-obedience, while weakly improving her own payoff. Therefore, there must be an optimal full-commitment outcome that is comonotone, which is also stable by \autoref{theorem-cyclical} and \autoref{lemma: comonotone}.

\section{Application: The Market for Lemons} \label{section: applications}

A classic insight from \citet*{Akerlof1970} is that in markets with asymmetric information, adverse selection can lead to dramatic efficiency loss. In practice, buyers and sellers often rely on warranty or third-party certification to overcome this inefficiency. A seemingly more direct solution to their predicament is for the seller to fully reveal her private information, so that there is no information asymmetry between players.  In this section, however, we show that this apparently easy fix to the adverse selection problem relies on unrealistic assumptions on the seller's ability to commit. Indeed, we show that any information disclosure that improves efficiency cannot be credible.

To fix ideas, we adapt the formulation in \cite{mas1995microeconomic} and consider a seller who values an asset she owns (say, a car) at $\theta \in \Theta\subset [0,1]$; two buyers (1 and 2) both value the car at $v(\theta)$  which is increasing in $\theta$. Buyers share a common prior belief $\mu_{0}\in\Delta(\Theta)$. We assume $v(\theta)>\theta$ for all $\theta\in \Theta$ so there is common knowledge of gain from trade. Moreover, we assume $E_{\mu_0}[v(\theta)]<1$, so that without information disclosure, some cars will not be traded due to adverse selection. Below we first describe the base game without information disclosure, then augment the base game to allow the seller to choose a test to influence the buyers' belief.

\paragraph{The Base Game $G$:} The seller and the buyers move simultaneously. The seller learns her value and chooses an ask price $a_s\in A_S=[0,v( 1)]$; each buyer $i=1,2$ chooses a bid $b_i\in A_i = [0, v(1) ]$. If the ask price is lower than or equal to the highest bid, the car is sold at the highest bid to the winning buyer, and ties are broken evenly. If the ask price is higher than the highest bid, the seller keeps the car and receives the reserve value $\theta$, while both buyers get $0$. More formally, the seller's payoff function is
\[u_S(\theta,a_S,b_1,b_2)=\begin{cases}
\max\{b_1,b_2\}&\text{ if }a_S\leq \max\{b_1,b_2\}\\
\theta&	\text{ if }a_S> \max\{b_1,b_2\}
\end{cases}
\]
and buyer $i$'s payoff is
\[u_i(\theta,a_S,b_1,b_2)=\begin{cases}
v(\theta)-b_i&\text{ if }b_i>b_{-i}\text{ and }b_i\geq a_S\\
\frac{1}{2}	[v(\theta)-b_i]&\text{ if }b_i=b_{-i}\text{ and }b_i\geq a_S\\
0&\text{ otherwise.}
\end{cases}
\]

\paragraph{The Game with Disclosure:}
Before the base game is played, the seller chooses a test $\lambda$ to publicly disclose information to the buyers.\footnote{In our setting, $\lambda$ determines only the buyers' information structure, and the seller is perfectly informed about $\theta$. That is, the seller 
cannot prevent herself from learning the true quality of the car. This differs from \citet{kartik2019lemonade}, who fully characterize payoffs in the market for lemons under all possible information structures.}
Together the test $\lambda$ and the base game $G$ defines a Bayesian game $\langle G,\lambda \rangle$. Every message $m$ from the test $\lambda$ induces a posterior belief $\mu_m \equiv \lambda(\cdot|m) \in \Delta(\Theta)$ for the buyers. The buyers $i=1,2$ choose their respective bids $\beta_i(m)$, while the seller choose an ask price $\alpha_S(\theta,m)$. We restrict attention to Bayesian Nash equilibria where the seller plays her \emph{weakly dominant strategy} $\alpha_S(\theta,m)=\theta$.  As we show in \cref{lemma: BNEexist}, such equilibria exist in $\langle G,\lambda\rangle$ for every $\lambda$. These equilibria also give rise to the familiar fixed-point characterization of equilibrium price:  buyers' bids satisfy
\begin{equation*}
\max \big\{ \beta_1(m), \beta_2(m) \big\} =  E_{\mu_m} \big[v(\theta)|\theta\leq \max\big\{\beta_1(m),\beta_2(m) \big\} \big].
\end{equation*}

The trading game above differs from the Sender-Receiver setting in \cref{section: model} in two ways: first, the Sender in the current setting publicly discloses information to multiple Receivers; second, in addition to the Receivers, the Sender also chooses an action (ask price) after observing the realization of the test. Nevertheless, the notion of stable outcome distribution extends to the current setting. In particular, the credibility notion is based on the same idea that the Sender cannot profitably deviate to a different test without changing the message distribution. The Receiver Incentive Compatible (R-IC) condition, meanwhile, is replaced by a new IC condition that asks both the Sender and Receivers to play according to a Bayesian Nash equilibrium in $\langle G,\lambda \rangle$. 
As mentioned above, in the market for lemons we will focus on a special class of 
Bayesian Nash equilibria where the seller plays her \textit{weakly dominant strategy} $\alpha_S(\theta,m)=\theta$ in the game $\langle G,\lambda\rangle$. We will call such profiles $(\lambda,\sigma)$ WD-IC to distinguish from the weaker IC requirement. The formal discussion of our credibility notion in this multiple-Receiver setting is notationally cumbersome, and is deferred to \cref{section: games}.

Next we state our result,  discuss its implications, and provide intuition for its proof. As a benchmark, fix an arbitrary message ${m}_0\in M$, and let $\lambda_0 \equiv \mu_{0}\times \delta_{{m}_0}$ be a null information structure. Let $R_0$ denote the supremum of the seller's payoffs among profiles $(\lambda_0, \sigma)$ that are WD-IC, so $R_0$ represents the highest equilibrium payoff the seller can achieve when providing no information.

\begin{proposition}\label{prop: adverse selection}
Under every credible and WD-IC profile, the seller's payoff is no more than $R_0$.
\end{proposition}

\autoref{prop: adverse selection} implies that any information that can be credibly disclosed is not going to improve the seller's payoff compared to the no-information benchmark. This is in sharp contrast to the full-commitment case, where the seller would like to fully reveal the car's quality, and all car types $\theta$ are sold at $v(\theta)$, which would allow the seller to capture all surplus from trade.

Let us describe the intuition behind the proof for \autoref{prop: adverse selection}. For each message $m$ from the seller's test $\lambda$, let $\Theta(m)$ denote the support of the buyer's posterior belief after observing $m$. A key step in proving \autoref{prop: adverse selection} is to show that there exists a common trading threshold $\tau$ such that for each message $m$, a car of quality $\theta\in \Theta(m)$ is traded if and only if $\theta\le \tau$. To see why, suppose towards a contradiction that the trading threshold in message $m$ is higher than the threshold in another message $m'$. We show in the proof that the seller would then have a profitable deviation by swapping some of the cars slightly below the higher threshold in message $m$ with an equal amount of cars from $m'$ that are slightly above its lower threshold.\footnote{This deviation is profitable because it allows the seller to replace the higher-quality cars traded in $m$ with the lower-quality, untraded cars in $m'$. After this swapping, the lower-quality cars are now sold at the price for the higher-quality cars in $m$, while the higher-quality cars are now retained by the seller in $m'$.} Because this deviation does not change the seller's message distribution, it is also undetectable. Therefore, credibility demands a common threshold $\tau$ that applies across messages. Given this common threshold $\tau$, we then apply Tarski's fixed-point theorem to show that when no information is disclosed, there is an equilibrium that features a higher trading threshold $\tau'\geq \tau$. Since a higher threshold means more cars are being traded, which in turn increases the seller's payoff, the seller's payoff under every stable outcome is therefore weakly dominated by her payoff from a no-information outcome, and this proves our result.

\section{Discussion} \label{section: discussions}

\subsection{Relationship to \citet{rochet1987necessary}} \label{section: relationship}

The $u_S$-cyclical monotonicity condition in our characterization closely resembles the cyclical monotonicity condition for implementing transfers in \cite{rochet1987necessary}. The reader might wonder why cyclical monotonicity arises in our setting despite the lack of transfers. The connection is best summarized by the following three equivalent conditions from optimal transport theory (see, for example, Theorem 5.10 of \citet{villani}).

\begin{thmot}
Suppose $X$ and $Y$ are both finite sets, and $u:X \times Y \rightarrow \mathbb{R}$ is a real-valued function. Let $\mu$ be a probability measure on $X$ and $\nu$ be a probability measure on $Y$, and $\Pi(\mu,\nu)$ be the set of probability measures on $X\times Y$ such that the marginals on $X$ and $Y$ are $\mu$ and $\nu$, respectively. Then for any $\pi^*\in \Pi(\mu,\nu)$, the following three statements are equivalent:
\begin{enumerate}[itemsep=0em, topsep=0.3em]
        \item $\pi^*\in \argmax_{\pi\in\Pi(\mu ,\nu)} \sum_{x,y} \pi(x,y)u(x,y)$;
        \item $\pi^*$ is $u$-cyclically monotone. That is, for any $n$ and $(x_1,y_1),...,(x_n,y_n)\in \supp(\pi^*)$,
        \[\sum_{i=1}^n u(x_i,y_i)\geq \sum_{i=1}^n u(x_i,y_{i+1})\]
    \item There exists $\psi:Y\rightarrow \mathbb{R}$ such that for any $(x,y)\in\supp(\pi^*)$ and any $y'\in Y$,\footnote{This statement can also be equivalently written as: there exists $\phi:X\rightarrow \mathbb{R}$ and $\psi:Y\rightarrow \mathbb{R}$, such that $\phi(x)+\psi(y)\geq u(x,y)$ for all $x$ and  $y$, with equality for $(x,y)$ in the support of $\pi^*$.}
    \[
    u(x,y)-\psi(y)\geq u(x,y')-\psi(y').
    \]

\end{enumerate}
\end{thmot}

Our \autoref{theorem-cyclical} builds on the equivalence between 1 and 2 in the Kantorovich duality theorem above to show the equivalence between credibility  and $u_S$-cyclical monotonicity. 

\citet{rochet1987necessary}'s classic result on implementation with transfers follows from the equivalence between 2 and 3. To see this, consider a principal-agent problem where the agent's private type space is $\Theta$ with full-support prior $\mu_0$, and the principal's action space is $A$. The agent's  payoff is $u(\theta,a)-t$, where $t$ is the transfer she makes to the principle. Given an allocation rule $q:\Theta\rightarrow A$, let $v_q(\theta,\theta') \equiv u(\theta,q(\theta'))$ denote the payoff that a type-$\theta$ agent obtains from the allocation intended for type $\theta'$. Let $X=Y = \Theta$ and $\mu = \nu = \mu_0$ in the Kontorovich Duality theorem above, and consider the distribution $\pi^*\in \Pi(\mu,\nu)$ defined by 
\begin{equation*}
\pi^*(\theta,\theta')=
\begin{cases}
\mu_0(\theta) & \text{ if } \theta=\theta'\\
0 & \text{ otherwise} 
\end{cases}
\end{equation*}
By the equivalence of 2 and 3 in the Kantorovich duality theorem, $\pi^*$ is $v_q$-cyclically monotone if and only if there exists $\psi:\Theta\rightarrow \mathbb{R}$ such that for all $\theta,\theta'\in\Theta$,
$v_q(\theta,\theta)-\psi(\theta)\geq v_q(\theta,\theta')-\psi(\theta')$. That is,
\[u(\theta,q(\theta))-\psi(\theta)\geq u(\theta,q(\theta'))-\psi(\theta'),\]
so the allocation rule $q$ can be implemented by the transfer rule $\psi:\Theta\rightarrow \mathbb{R}$. The $v_q$-cyclical monotonicity condition says that for every sequence $\theta_1,...,\theta_n\in\Theta$ with $\theta_{n+1}\equiv \theta_1$,
\[
\sum_{i=1}^n u(\theta_i,q(\theta_i))\geq \sum_{i=1}^n u(\theta_i,q(\theta_{i+1})).
\]
This is exactly the cyclical monotonicity condition in \citet{rochet1987necessary}.

When $X=\Theta$ is interpreted as the set of an agent's true types and $Y=\Theta$ interpreted as the set of reported types, the distribution $\pi^*$ constructed in the previous paragraph can be interpreted as the agent's truthful reporting strategy. Based on this interpretation, \citet{rahman2010detecting} uses the duality between 1 and 3 to show that the incentive compatibility of truthful reporting subject to quota constraints is equivalent to implementability with transfers.

\subsection{Finite-Sample Approximation}\label{section: finite sample}

As discussed in \cref{section: modelsetup}, we interpret our model as one where the Sender designs a test that assigns scores to a large population of realized $\theta$'s; in particular, our model abstracts away from sampling variation, so there is no uncertainty in the population's realized type distribution. In this section we explicitly allow sampling variation by considering a finite-sample model where the Sender observes $N$ random i.i.d. draws from $\Theta$, and  assigns each realized $\theta$ a message $m\in M$, subject to certain message quotas---in particular, these message quotas substitute for the Sender's commitment to message distributions in the continuum model. We will show that credible and R-IC profiles in our continuum model are approximated by credible and R-IC profiles in the discrete model when sample size is large.

Consider a finite i.i.d sample of size $N$ drawn from the type space $\Theta$ according to the prior distribution $\mu_0$. The set of all possible empirical distributions over $\Theta$ in this $N$-sample captures the sampling variations in the realized type distribution, and can be written as
\[
\mathcal{F}_\Theta^N=\Big\{{f}/{N}:f\in \mathbb{N}^{|\Theta|},\sum_{\theta \in \Theta} f(\theta)=N\Big\}.
\]
The Sender assigns each realized $\theta$ in the $N$-sample a message $m\in M$, which leads to an $N$-sample of messages. Let
\[
\mathcal{F}_M^N= \Big\{{f}/{N}:f\in \mathbb{N}^{|M|},\sum_{m\in M} f (m) = N \Big\}
\]
denote the set of $N$-sample empirical distributions over messages. Lastly, for a pair of state and message distributions $(\mu, \nu)$, let
\begin{equation*}
X^N(\mu,\nu)=\Big \{{f}/{N}:f\in \mathbb{N}^{|\Theta|\times |M|}, \sum_{\theta} f(\theta,\cdot) = N \nu(\cdot), \sum_{m} f(\cdot,m)=N \mu (\cdot) \Big \}
\end{equation*}
denote the set of $N$-sample empirical joint distributions over states and messages that have marginals $\mu$ and $\nu$. Notice that $X^N(\mu,\nu)\neq \emptyset$ if and only if $\mu\in \mathcal{F}_\Theta^N$ and $\nu\in \mathcal{F}_M^N$.\footnote{Notice that for any $f\in\mathbb{N}^{|\Theta|\times |M|}$, the sum of any row or column has to be integer, so $X^N(\mu,\nu)=\emptyset$ if either $\mu\notin \mathcal{F}_\Theta^N$ or $\nu\notin \mathcal{F}_M^N$. On the other hand if $\mu\in F_{\Theta}^N$ and $\nu\in F_M^N$, a $\lambda\in X^N(\mu,\nu)$ can be constructed from the so-called Northwest corner rule.}

Let us now define the $N$-sample analogue of credible and R-IC profiles. We consider a Sender who assigns a message $m\in M$ to each realized $\theta\in \Theta$ subject to a message quota $\nu^N \in \mathcal{F}_M^N$. An $N$-sample profile is therefore a triple $(\nu^N, \phi^N, \sigma^N)$, where $\phi^N: \mathcal{F}_\Theta^N  \rightarrow \Delta(\Theta\times M)$ is a Sender's strategy that takes every realized empirical distribution over states $\mu^N\in \mathcal{F}^N_{\Theta}$ to a joint distribution $\phi^N(\mu^N) \in X^N(\mu^N,\nu^N)$; meanwhile, $\sigma^N: M\rightarrow A$ is a Receiver's strategy that assigns an action to each observed message.\footnote{Note that our formulation of the Sender's strategy assumes that the Sender conditions her strategy only on the empirical distribution of the realized $N$ samples, and ignores the identity of each individual sample point.}

The definitions of Sender credibility and Receiver incentive compatibility in the $N$-sample setting mirror those in our continuum model. In particular, we say an $N$-sample profile $(\nu^N, \phi^N, \sigma^N)$ is credible  if for each realized empirical distribution over $\Theta$, the Sender always chooses an optimal assignment of messages subject to the message quotas specified in $\nu^N$: $(\nu^N, \phi^N, \sigma^N)$ is credible if for every $\mu^N\in \mathcal{F}_{\Theta}^N$,
\begin{equation*}
    \phi(\mu^N)\in\argmax_{ \lambda^N \in X(\mu^N,\nu^N)}\sum_{\theta,m} \lambda^N(\theta,m) u_S(\theta,\sigma^N(m)).
\end{equation*}
We say the $N$-sample profile $(\nu^N, \phi^N, \sigma^N)$ is Receiver incentive compatible (R-IC) if $\sigma^N$ best-responds to the Sender's strategy $\phi^N$. In particular, let $P^N$ denote the probability distribution over $\mathcal{F}^N_{\Theta}$ induced by i.i.d. draws from the prior distribution $\mu_0\in \Delta(\Theta)$, and let $\phi^N(\theta,m|\mu^N )$ be the probability assigned to $(\theta,m)$ in the joint distribution $\phi^N(\mu^N)$ chosen by the Sender. The profile $(\nu^N, \phi^N, \sigma^N)$ is R-IC if
\begin{equation*}
    \sigma^N\in \argmax_{\sigma':M\rightarrow A} \sum_{\mu^N\in \mathcal{F}_{\Theta}^N} P^N(\mu^N)\sum_{\theta,m} \phi(\theta,m|\mu^N) u_R(\theta,\sigma'(m)).
\end{equation*}

\cref{proposition: finite} below shows that credible and R-IC profiles in the continuum model are approximated by credible and R-IC profiles in the $N$-sample model, provided $N$ is sufficiently large. 
Note that in the second statement in \cref{proposition: finite}, we distinguish a strictly credible profile $(\lambda^*,\sigma^*)$ in the continuum model as one where $\lambda^*$ is the unique maximizer in \cref{definition: credible}; similarly, $(\lambda^*,\sigma^*)$ is strictly R-IC if $\sigma^*$ is the unique maximizer in \cref{definition: RIC}. 

\begin{proposition}\label{proposition: finite}
\begin{enumerate}
\item Let $(\lambda^*,\sigma^*)$ be a profile in the continuum model. If for every $\varepsilon>0$, there exists a finite credible profile $(\nu^N, \phi^N, \sigma^N)$ for some sample size $N$, such that $|\nu^N-\lambda^*_M| < \varepsilon$, $|\sigma^N-\sigma^*|<\varepsilon$ and $P(|\phi^N(F_{\Theta}^N) - \lambda^* | < \varepsilon) > 1-\varepsilon$, then $(\lambda^*, \sigma^*)$ is credible and R-IC.
\item Suppose $(\lambda^*,\sigma^*)$ is a strictly credible and strictly R-IC profile in the continuum model, then for each $\epsilon>0$ there exists a finite-sample credible and R-IC profile $(\nu^N, \phi^N, \sigma^N)$ such that $|\nu^N - \lambda^*_M| <\varepsilon$, $|\sigma^N-\sigma^*|<\varepsilon$ and $P(|\phi^N(F_{\Theta}^N)-\lambda^*| < \varepsilon)>1-\varepsilon$.
\end{enumerate}
\end{proposition}
The first statement in \cref{proposition: finite} is analogous to the upper-hemicontinuity of Nash equilibrium correspondences: if a profile $(\lambda^*,\sigma^*)$ in the continuous model can be arbitrarily approximated by credible and R-IC profiles in the finite model, then profile $(\lambda^*,\sigma^*)$ must itself be credible and R-IC.
Conversely, the second statement in \cref{proposition: finite} can be interpreted in a way similar to the lower-hemicontinuity of strict Nash equiliria: if a profile $(\lambda^*,\sigma^*)$ in the continuous model is strictly credible and strictly R-IC, then it can be arbitrarily approximated by credible and R-IC profiles in the finite model.

\subsection{Extensive-Form Foundation for Credibility} \label{section: extensivesingle}
Our solution concept analyzes credible persuasion through the lens of a partial-commitment model, which is motivated by the observability of the Sender's message distribution. To formalize this connection, we propose an extensive-form game between the Sender and the Receiver where the Receiver observes the message distribution of each test. We show that the set of pure-strategy subgame perfect Nash equilibria of this extensive-form game corresponds to the set of profiles $(\lambda,\sigma)$ that are credible, R-IC, and give the Sender higher than her worst no-information payoff. 

\medskip

Consider the following game between the Sender and the Receiver:

\begin{enumerate}[nolistsep]
	\item The Sender chooses a test $\lambda \in \Delta(\Theta\times M)$ which satisfies $\lambda_{\Theta} =\mu_0 $; 
	\item Nature draws a pair of state and message $(\theta,m)$ according to $\lambda$;
	\item The Receiver observes $m$, as well as the distribution of messages induced by $\lambda$, $\lambda_M\in\Delta(M)$, then chooses an action $a\in A$.
\end{enumerate}

The Sender's strategy set is $\Lambda\equiv \{\lambda\in \Delta(\Theta\times M): \lambda_{\Theta} = \mu_0 \} $, and the Receiver's strategy set is $\Xi=\{\rho:\Delta(M) \times M \rightarrow A \}$, where the first argument is the distribution of messages and the second argument is the message. The Sender's payoff is 
\begin{equation*}
U_S(\lambda,\rho)=\sum_\theta \sum_m \lambda(\theta,m)u_S\big(\theta,\rho(\lambda_M, m)\big )
\end{equation*}
while the Receiver's payoff is 
\begin{equation*}
U_R(\lambda,\rho)=\sum_\theta \sum_m \lambda(\theta,m)u_R \big(\theta,\rho(\lambda_M,m) \big)
\end{equation*}

The solution concept is pure-strategy subgame perfect Nash equilibrium (SPNE).\footnote{Strengthening the solution concept to PBE does not change our result.} Notice that in the extensive-form game above, after the Sender chooses a degenerate test that always sends the a single message, the decision node for Nature forms the initial node of a proper subgame. In fact, these are also the only proper subgames in the extensive-form game, where subgame perfection has bite.

We call a profile $(\lambda,\sigma)$, as defined in \cref{section: modelsetup}, a \textit{pure-strategy SPNE outcome} of the extensive-form game above if there exists a pure-strategy SPNE $(\lambda,\rho)$ of the  extensive-form game such that $\rho(\lambda_M,m)=\sigma(m)$ for all $m\in M$. The following result relates the SPNE of this extensive-form game to our solution concept.

\begin{proposition} \label{proposition: S-R extensive form}
A profile $(\lambda,\sigma)$ is a pure-strategy SPNE outcome of the extensive-form game if and only if
\begin{enumerate}
	\item $(\lambda,\sigma)$ is credible and R-IC; that is, $(\lambda,\sigma)$ satisfies \cref{definition: credible,definition: RIC}.
	\item The Sender's payoff from $(\lambda,\sigma)$ is greater than her lowest no-information payoff:
	\begin{equation*}
	\sum_\theta \sum_m \lambda(\theta,m)u_S(\theta,\sigma(m))\geq \min_{a \in A_0} \sum_\theta\mu_{0}(\theta)u_S(\theta,a),
	\end{equation*}
	where $A_0=\arg\max_{a\in A} \sum_\theta\mu_{0}(\theta)u_R(\theta,a)$ is the Receiver's best-response set to the prior belief $\mu_0$.
\end{enumerate}
\end{proposition}

\autoref{proposition: S-R extensive form} shows that if one focuses on profiles  where the Sender is weakly better off than disclosing no information (as we do in this paper), the extensive-form game rationalizes our definition for credibility. The reason that the Sender's payoff must be higher than her no-information payoff in the extensive-form game is that in any equilibrium, if she deviates to a no-information test $\lambda=\mu_0\times \delta_{m_0}$, the ensuing decision node forms the initial node of a proper subgame. Subgame perfection then demands that the Receiver plays a best response to his prior, which in turn ensures that the Sender obtains a no-information payoff following this deviation. Therefore, the Sender's equilibrium payoff must be weakly higher than her worst no-information payoff.

\section{Conclusion} \label{section: conclusion}

This paper offers a new notion of credibility for persuasion problems. We model a Sender who can commit to a test only up to the details that are observable to the Receiver. The Receiver does not observe the chosen test but observes the distribution of messages. This leads to a model of partial commitment where the Sender can undetectably deviate to tests that induce the same distribution of messages. Our framework characterizes when, given the Receiver's best response, the Sender has no profitable deviation. 

We show that this consideration eliminates the prospects for credible persuasion in settings with adverse selection. In other settings, we show that the requirement is compatible with the Sender still benefiting from persuasion. More generally, we show that our requirement translates to a cyclical monotonicity condition on the induced distribution of states and actions. 

In addition to the theoretical findings above, our work has several applied implications. The first is that the commitment assumption commonly made in the Bayesian persuasion literature may be innocuous in some applications. For example, \cite{xiang2021physicians} uses a Bayesian persuasion framework to empirically study information transmission in the physician-patient relationship, where the physician is assumed to commit to a recommendation policy that is observable to patients. But in practice, patients may observe the distribution of recommendations and not the recommendation policy. Our results imply that in her context, this is enough: knowing the distribution of recommendations suffices because both the physician and the patients' payoffs are supermodular, and the patients face a binary decision, so the optimal full-commitment policy is credible according to our \cref{proposition: no cost}.

Our work also speaks to why certain industries (such as education) can effectively disclose information by utilizing their own rating systems, while some other industries (such as car dealerships) must resort to other means to addressing asymmetric information, such as third-party certification or warranties. Our results provide a rationale: in industries that exhibit adverse selection, the informed party cannot credibly disclose information through its own ratings even if it wishes to do so.

The notion of credibility we consider in this paper is motivated by the observability of the Sender's message distribution. In some settings, the Receiver may observe more than the distribution of messages; for example, she may observe some further details of the test, such as how some states of the world map into messages. In other settings, the Receiver may observe less; e.g., she may see the average grade, but not its distribution. To capture these various cases, one would then formulate the problem of ``detectable'' deviations differently. We view it to be an interesting direction for future research to understand how different notions of detectability map into different conditions on the outcome distribution.

{\small
	\addcontentsline{toc}{section}{References}
	\setlength{\bibsep}{0.25\baselineskip}
	\bibliographystyle{jpe}
	\bibliography{marginalcommitment}
}

\appendix
\addtocontents{toc}{\protect\setcounter{tocdepth}{1}} 

\section{Appendix} \label{section: appendix}

\subsection{Proof of \autoref{theorem-cyclical}} \label{section: appendix proof of thm 1}

The following lemma, which will play a key role in the proof of \autoref{theorem-cyclical}, is a finite version of Theorem 5.10 of \cite{villani}.

\begin{lemma}\label{lemma: optimaltransport}
    Suppose both $X$ and $Y$ are finite sets, and $u:X\times Y\rightarrow \mathbb{R}$ is a real function. Let $p\in \Delta(X)$ and $q\in\Delta(Y)$ be two probability measure on $X$ and $Y$ respectively, and $\Pi(p,q)$ be the set of joint probability measure on $X\times Y$ such that the marginals on $X$ and $Y$ are $p$ and $q$. The following two statements are equivalent:
    \begin{enumerate}[nolistsep]
        \item $\pi^*\in \argmax_{\pi\in\Pi(p,q)} \sum_{x,y} \pi(x,y)u(x,y)$;
        \item $\pi^*$ is $u$-cyclically monotone. That is, for any $n$ and $(x_1,y_1),...,(x_n,y_n)\in \supp(\pi^*)$,
        \[\sum_{i=1}^n u(x_i,y_i)\geq \sum_{i=1}^n u(x_i,y_{i+1})\]
    \end{enumerate}
    where $y_{n+1}\equiv y_1$.
\end{lemma}

\begin{proof}[Proof of \autoref{theorem-cyclical}]
For the ``if'' direction, suppose $\pi$ is $u_R$-obedient, $u_S$-cyclically monotone, and satisfies $\pi_{\Theta} = \mu_{0}$. The proof is by construction.

Since $\pi_{\Theta} = \mu_{0}$, we can construct a test $(M,\lambda^*)$ by setting $M= A$ and $\lambda^*=\pi$; furthermore, let $\sigma^*$ be the identity map from $M$ to $A$. It is straightforward to see that the profile $(\lambda^*,\sigma^*)$ induces the outcome distribution $\pi$. We show that $(\lambda^*,\sigma^*)$ is credible. First, since $\pi$ is $u_R$-obedient, we have that for each $a\in A$,
	\begin{equation*}
 	a  \in \argmax_{a'} \sum_\Theta  u_R(\theta,a')\; \pi(\theta, a ).
 	\end{equation*}
Since $\sigma^*$ is an identity map, it follows that for each $m\in M$, 
	\begin{equation*}
 	\sigma^*(m)  \in \argmax_{a'} \sum_\Theta u_R(\theta,a')\;\pi(\theta, \sigma^*(m) ).  
 	\end{equation*}
Furthermore, since $\lambda^*=\pi$ and $\sigma^*$ is injective, we have $\lambda^*(\theta,m) = \pi(\theta,\sigma^*(m))$ for all $\theta \in \Theta$ and $m \in M$. So
\begin{equation*}
	\sigma^*\in \argmax_{\sigma: M \rightarrow A} \; \sum_{\Theta\times M} u_R(\theta,\sigma(m))\;\lambda^*(\theta,m),
\end{equation*}
which means $\sigma^*$ is a  best response to $\lambda^*$.

It remains to show that the Sender does not benefit from choosing any other test in  $\Lambda(\mu_{0}, \lambda^*_M)$. Observe that since $\pi$ is $u_S-$cyclically monotone, every sequence $(\theta_1,a_1),\ldots,(\theta_n,a_n)$ in $\supp(\pi)$ where $a_{n+1}\equiv a_1$ satisfies
	\begin{equation*}
	\sum_{i=1}^n u_S(\theta_i,a_i)\geq \sum_{i=1}^n u_S(\theta_i,a_{i+1}).
	\end{equation*}
Since $\lambda^* = \pi$ and $\sigma^*$ is the identity mapping, this further implies 
	\begin{equation*}
	\sum_{i=1}^n u_S(\theta_i, \sigma^*(m_i))\geq \sum_{i=1}^n u_S(\theta_i, \sigma^*(m_{i+1}));
	\end{equation*}
for every sequence $(\theta_1,m_1),\ldots,(\theta_n,m_n)\in \supp(\lambda^*)$ with $m_{n+1} = m_1$. In addition, $\lambda^*_{\theta} = \mu_{0}$ and $\lambda^*_{M} = \lambda^*_{M}$ by construction. By \autoref{lemma: optimaltransport}, $\lambda^*$ satisfies
	\begin{equation*}
	\lambda^* \in \argmax_{\lambda\in \Lambda(\mu_{0},\lambda^*_M)} \; \sum_{\Theta \times M}\; u_S(\theta,\sigma(m))\;\lambda(\theta,m)
	\end{equation*}
which means $\lambda^*$ is Sender optimal conditional on its message distribution.
\bigskip

For the ``only if'' direction, suppose $\pi$ is stable and thus induced by a credible and R-IC profile $(\lambda^*,\sigma^*)$. Since $\sigma^*$ best responds to the messages from $\lambda^*$, the $u_R$-obedience of $\pi$ follows from \cite{BM2016}. It remains to show that $\pi$ is $u_S$-cyclical monotone. Suppose by contradiction that $\pi$ is not $u_S$-cyclically monotone, which implies that there exists a sequence $(\theta_1,a_1),\ldots,(\theta_n,a_n)\in \supp(\pi)$ such that 
\begin{equation*} 
\sum_{i=1}^n u_S(\theta_i,a_i)<\sum_{i=1}^n u_S(\theta_i,a_{i+1}),
\end{equation*}
where $a_{n+1}=a_1$. Since $\pi$ is induced by $(\lambda^*,\sigma^*)$, for each $i=1,\ldots,n$ there exists $m_i$ such that $m_i\in \sigma^{*-1}(a_i)$ and $(\theta_i,m_i)\in \supp(\lambda^*)$, so we have a sequence $(\theta_1,m_1),\ldots,(\theta_n,m_n)\in \supp(\lambda^*)$ that satisfies
\begin{equation} \label{eq-sc}
	\sum_{i=1}^n u_S(\theta_i,\sigma^*(m_i))<\sum_{i=1}^n u_S(\theta_i,\sigma^*(m_{i+1})),
\end{equation}
where $m_{n+1}=m_1$. Define $v(\theta,m) \equiv u_S(\theta,\sigma^*(m))$. Since $(\lambda^*,\sigma^* )$ is credible, we have
\begin{equation*}
	\lambda^* \in \argmax_{\lambda\in\Lambda(\mu_{0},\lambda^*_M)} \sum_{\Theta \times M} v(\theta,m) \lambda(\theta,m).
\end{equation*}
\autoref{lemma: optimaltransport} implies that $\lambda^*$ is $v$-cyclically monotone.
Since $(\theta_1,m_1),\ldots,(\theta_n,m_n)$ is in  $\supp(\lambda^*)$, the $v$-cyclical monotonicity of $\lambda^*$ implies
\begin{equation*}
\sum_{i=1}^n u_S(\theta_i,\sigma^{*}(m_i))\geq \sum_{i=1}^n u_S(\theta_i,\sigma^{*}(m_{i+1}))
\end{equation*}
where $m_{n+1}=m_1$, which is a contradiction to \eqref{eq-sc}. So $\pi$ must be $u_S$-cyclically monotone.

\end{proof}

\subsection{Proof of \autoref{lemma: comonotone}}

\begin{lemma} \label{lemma: index crossing}
Let $t:\{1,\ldots,n\}\rightarrow\{1,\ldots,n\} $ be a bijection. Suppose $t$ is not the identity mapping, then there exists $k$ such that $t(k)>k$ and $t(t(k))<t(k)$.
\end{lemma}
\begin{proof}
Suppose by contradiction that for every $k$ such that $t(k)>k$, $t(t(k))\geq t(k)$. Notice that since $t$ is a bijection, $t(t(k))\neq t(k)$ (otherwise $t(k)=k$ contradicting $t(k)>k$), so for every $k$ such that $t(k)>k$, $t(t(k))\geq t(k)+1$.

Since $t$ is not the identity mapping, there exists $k_1$ such that $t(k_1)>k_1$ or equivalently $t(k_1)\geq k_1+1$. Define iteratively that $k_j=t(k_{j-1})$ for $j=2,\ldots,n$, we have $k_j-k_{j-1}\geq 1$. Then we have $k_n\geq k_1+n>n$, which is contradiction. So there exists  $k$ such that $t(k)>k$ and $t(t(k))<t(k)$.
\end{proof}

First, we show that comonotonicity implies $u_S$-cyclical monotonicity when $u_S$ is supermodular. Suppose an outcome distribution $\pi \in \Delta(\Theta\times A )$ is comonotone, then $\supp(\pi)$ is totally ordered. Take any sequence $(\theta_1,a_1),\ldots,(\theta_n,a_n)\in \supp(\pi)$ and assume without loss of generality that  $(\theta_{i},a_{i})$ is increasing in $i\in\{1,\ldots,n\}$. We will show that for any permutation $t:\{1,\ldots,n\}\rightarrow \{1,\ldots,n\}$,

\begin{equation*}
u_S(\theta_1,a_1)+\ldots+u_S(\theta_n,a_n)\geq u_S(\theta_1,a_{t(1)})+\ldots+u_S(\theta_n,a_{t(n)}),
\end{equation*}
which then proves the statement. In particular, for each permutation $t$, let $v(t) \equiv u_S(\theta_1,a_{t(1)})+\ldots+u_S(\theta_n,a_{t(n)})$ denote the value obtained from summing $u_S$ according to the state-action pairings in $t$ and let $I$ denote the identity map. We show that $v(I) \ge v(t)$ for every permutation $t$.

To this end, take any permutation $t$ that is not an identity mapping, and let $l(t)$ denote the number of fixed points of $t$ (which may be zero). By \cref{lemma: index crossing}, there exists $k^*$ such that $t(k^*)>k^*$ and $t(t(k^*))<t(k^*)$. The supermodularity of $u_S$ implies 
\begin{equation} \label{equation: rearrangement}
u_S(\theta_{t(k^*)},a_{t(k^*)})+u_S(\theta_{k^*},a_{t(t(k^*))})\geq u_S(\theta_k,a_{t(k^*)}) + u_S(\theta_{t(k^*)},a_{t(t(k^*))}).
\end{equation}
Define a new permutation $\hat{t}$ so that $k$ is mapped to $t(t(k))$ while $t(k)$ is mapped to $t(k)$, while all other pairings remain unchanged. Formally,
\begin{equation*}
	\hat{t}(k) =
	\begin{cases}
	{t}(k) & \text{for all } k \ne k^*, t(k^*) \\
	t(t(k^*)) & \text{if } k=k^* \\
	t(k^*) & \text{if } k = t(k^*)
	\end{cases}
\end{equation*}
By \eqref{equation: rearrangement}, we have
\begin{equation*}
u_S(\theta_1,a_{\hat{t}(1)})+\ldots+u_S(\theta_n,a_{\hat{t}(n)})  \geq u_S(\theta_1,a_{t(1)})+\ldots+u_S(\theta_n,a_{t(n)}),
\end{equation*}
so we have constructed another permutation $\hat{t}$ with $v(\hat{t}) \ge v(t)$ and $l(\hat{t}) = l(t)+1$. Each time we iterate the process above, $v(.)$ weakly increases while the number of fixed points increases by one. Since $n<\infty$, the iteration terminates at the identity map $I$, so $v(I) \ge v(t)$ for every permutation $t$.
\bigskip

Next, suppose $u_S$ is strictly supermodular. We will show that $u_S$-cyclical monotonicity implies comonotonicity. Towards a contradiction, suppose that an outcome distribution $\pi$ is $u_S$-cyclically monotone but not comonotone. Then there exists $(\theta,a),(\theta',a')\in \supp(\pi)$ such that $\theta<\theta'$, $a>a'$. Since $u_S$ is strictly supermodular, 
	\begin{equation*}u_S(\theta,a)+u_S(\theta',a')< u_S(\theta,a')+u_S(\theta',a)\end{equation*}
	which is a contradiction to the $u_S$-cyclically monotonicity of $\pi$ when $(\theta_1,a_1) = (\theta,a)$ and $(\theta_2,a_2) = (\theta',a')$.

\subsection{Proof of \autoref{proposition: noinformation}}
Let $\pi$ be a stable outcome distribution, and suppose by contradiction that there exists two distinct actions $a_1,a_2\in \supp(\pi_a)$, say $a_1<a_2$. Let $I_1 \equiv  \{\theta\in\Theta|\pi(\theta,a_1)>0\}$ and $I_2 \equiv \{\theta\in\Theta|\pi(\theta,a_2)>0\}$ be the states associated with $a_1$ and $a_2$ in the support of $\pi$, respectively. 
By \cref{theorem-cyclical}, since $\pi$ is stable, it must be $u_R$-obedient, which implies
\begin{equation}\label{eq-obedient}
\sum_{\theta\in I_1}[u_R(\theta,a_1)-u_R(\theta,a_2)]\frac{\pi(\theta,a_1)}{\pi_a(a_1)}\geq 0\geq \sum_{\theta' \in I_2}[u_R(\theta',a_1)-u_R(\theta',a_2)]\frac{\pi(\theta',a_2)}{\pi_a(a_2)}
\end{equation}

Furthermore, since $u_S$ is strictly supermodular, $\pi$ is also comonotone by \autoref{theorem-cyclical} and \autoref{lemma: comonotone}, so any $\theta\in I_1$ and $\theta'\in I_2$ satisfies $\theta\leq \theta'$. Since $u_R$ is submodular, we have $u_R(\theta,a_1)-u_R(\theta,a_2) \leq u_R(\theta',a_1)-u_R(\theta',a_2)$ for all $\theta\in I_1$ and $\theta'\in I_2$, which implies 
\begin{equation*}
\max_{\theta\in I_1} \big\{ u_R(\theta,a_1)-u_R(\theta,a_2) \big\} \leq \min_{\theta' \in I_2} \big\{ u_R(\theta',a_1)-u_R(\theta',a_2) \big\}.	
\end{equation*}
So
\begin{equation}\label{eq-reverse-obedient}
\begin{aligned}
    \sum_{\theta\in I_1}[u_R(\theta,a_1)-u_R(\theta,a_2)]\frac{\pi(\theta,a_1)}{\pi_a(a_1)}&\leq \max_{\theta\in I_1} \big\{ u_R(\theta,a_1)-u_R(\theta,a_2) \big\} \\
    &\leq \min_{\theta' \in I_2} \big\{ u_R(\theta',a_1)-u_R(\theta',a_2) \big\} \\
    &\leq \sum_{\theta'\in I_2}[u_R(\theta',a_1)-u_R(\theta',a_2)]\frac{\pi(\theta',a_2)}{\pi_a(a_2)}
\end{aligned}
\end{equation}
Combining \eqref{eq-obedient} and \eqref{eq-reverse-obedient},  we have
\begin{equation*}
	\sum_{\theta\in I_1}[u_R(\theta,a_1)-u_R(\theta,a_2)]\frac{\pi(\theta,a_1)}{\pi_a(a_1)}= \max_{\theta\in I_1} \big\{ u_R(\theta,a_1)-u_R(\theta,a_2) \big\} =0
\end{equation*}
and 
\begin{equation*}
	\sum_{\theta' \in I_2}[u_R(\theta',a_1)-u_R(\theta',a_2)]\frac{\pi(\theta',a_2)}{\pi_a(a_2)}= \min_{\theta' \in I_2} \big\{ u_R(\theta',a_1)-u_R(\theta',a_2) \big\} =0
\end{equation*}
So $u_R(\theta,a_1)=u_R(\theta,a_2)$ for all $\theta\in I_1\cup I_2$. 

Since the argument above works for any $a_1,a_2\in \supp(\pi_a)$, it implies $u_R(\theta,a)=u_R(\theta,a')$ for all $\theta\in \Theta$ and all $a,a'\in \supp(\pi_a)$. However, this is a contradiction since by assumption, there exists no $a,a'\in A$ such that $a\neq a'$ and $u_R(\theta,a)=u_R(\theta,a')$ for all $\theta$.

Therefore $\supp(\pi_a)$ must be a singleton, denoted by $a^*$. Then $u_R$-obeidence implies $a^*\in \argmax_{a\in A} \sum_\theta\mu_0(\theta)u(\theta,a)$. So $\pi$ is a no-information outcome.

\subsection{Proof of \autoref{proposition: alignedpreference}}

\begin{proof}[Proof of statement 1]
    For each $a\in A$, let 
    \[P_a \equiv \{\mu\in\Delta(\Theta)|\sum_\theta \mu(\theta)u_R(\theta,a)>\sum_\theta \mu(\theta)u_R(\theta,a'),\forall a'\neq a\}\] which denotes the set of beliefs such that $a$ is the Receiver's strict best response. We prove our claim under the assumption that there exists $a^\circ \in A$ such that $\mu_{0} \in P_{a^\circ}$ (i.e. $a^\circ$ is the unique best response to $\mu_0$). Later we will show that this assumption holds for generic priors.

When the Sender's test is uninformative, the Receiver best responds to the Sender's messages by choosing $a^\circ$. The Sender's payoff is 
\begin{equation*}
v_0\equiv\sum_{\theta\in \Theta} \mu_0(\theta) u_S(\theta, a^0).	
\end{equation*}
We will show that there exists a stable outcome distribution that gives the Sender a higher payoff than $v_0$.

We consider the case where the sender benefits from persuasion, so $a^\circ \neq \overline{a}$, otherwise the Receiver is already choosing the sender's favourite action under the prior. For $\epsilon$ sufficiently small, consider the outcome distribution $\pi^{\epsilon}\in\Delta(\Theta\times A)$ defined by
\begin{equation*}
	\pi^{\epsilon}(\theta,a) =
	\begin{cases}
		\mu_{0}(\theta) & \text{if } \theta \ne \overline{\theta}, a = a^{\circ} \\
		\mu_{0}(\overline{\theta}) - \epsilon & \text{if } (\theta,a)  = (\overline{\theta}, a^{\circ}) \\
		\epsilon & \text{if } (\theta,a)  = (\overline{\theta}, \overline{a}) \\
		0 & \text{otherwise .}
	\end{cases}
\end{equation*}

We will show that for $\epsilon$ sufficiently small, $\pi^\epsilon$ is stable and gives the Sender higher payoff than $v_0$.

It can be easily seen that the support of $\pi^{\epsilon}$ is comonotone. Since $u_S$ is supermodular, $\pi^\epsilon$ is $u_S$-cyclically monotone by \cref{lemma: comonotone}.

Next we verify that for $\epsilon$ sufficiently small, $\pi^{\epsilon}$ satisfies $u_R$-obedience at the two actions $\{\overline{a}, a^\circ\}$. For $a^\circ$, note that since $\mu_{0}\in P_{a^\circ}$, we have
\begin{equation*}
\sum_{\theta\in \Theta} \mu_{0}(\theta)u(\theta,a^\circ)  > \sum_{\theta\in \Theta} \mu_{0}(\theta)\pi(\theta,a')  \text{ for all } a'\in A,
\end{equation*}
so for $\epsilon$ sufficiently small, 
\begin{equation*}
\sum_{\theta\in \Theta} \mu_{0}(\theta)u(\theta,a^\circ)-\epsilon u( \overline{\theta},a^\circ) \geq \sum_{\theta\in \Theta} \mu_{0}(\theta)\pi(\theta,a')-\epsilon u(\overline{\theta},a') \text{ for all } a'\in A.
\end{equation*}
which means $\pi^\epsilon$ satisfies $u_R$-obedience at $a^\circ$.

For $\overline{a}$, note that since every Receiver action is a best response to some belief (recall that this was assumed without loss of generality as explained in \autoref{section: whenisrestrictive}), there exists $\overline{\mu}\in\Delta(\Theta)$ such that $\overline{a} \in \arg\max_a \sum_\theta \overline{\mu}(\theta)u_R(\theta,a)$. So for every $a'\ne \overline{a}$,
\begin{equation*}
\sum_\theta \overline{\mu}(\theta)[u_R(\theta,\overline{a})-u_R(\theta,a')]\geq 0
\end{equation*}
Since $u_R$ is supermodular, $u_R(\theta,\overline{a})-u_R(\theta,a')$ is increasing in $\theta$, so if a belief $\mu'$ first order stochastically dominates $\overline{\mu}$, then
\begin{equation*}
\sum_\theta \mu'(\theta)[u_R(\theta,\overline{a})-u(\theta,a')]\geq \sum_\theta \overline{\mu}(\theta)[u_R(\theta,\overline{a})-u(\theta,a')]\geq 0 \text{ for all } a'\ne \overline{a}.
\end{equation*}
In particular, the Dirac measure $\delta_{\overline{\theta}}$ first order stochastically dominates $\overline{\mu}$, so the inequality above implies
\begin{equation*}
u_R(\overline{\theta},\overline{a})- u_R(\overline{\theta},a') \ge 0 \text{ for all } a'\ne \overline{a}.
\end{equation*}
So $\overline{a} \in \argmax_a u_R( \overline{\theta},a)$, and $\pi^\epsilon$ is $u_R$-obedient at action $\overline{a}$.
\bigskip

Finally, we show that the Sender obtains higher payoff from $\pi^\epsilon$ than $v_0$. Note that since by our assumption, $u_S(\overline{\theta},a')< u_S(\overline{\theta},\overline{a})$ for all $a'\ne \overline{a}$, we have
\begin{align*}
\sum_{\theta,a} \pi^{\epsilon}(\theta,a)u_S(\theta,a)&=\sum_{\theta \ne \overline{\theta}}	\mu_{0}(\theta)u_S(\theta,a^\circ)+(\mu_{0}(\overline{\theta})-\epsilon)u_S( \overline{\theta},a^\circ)+\epsilon u_S(\overline{\theta},\overline{a})\\
&>\sum_{\theta \ne \overline{\theta}}	\mu_{0}(\theta)u_S(\theta,a^\circ)+(\mu_{0}(\overline{\theta})-\epsilon)u_S( \overline{\theta},a^\circ)+\epsilon u_S(\overline{\theta},{a}^\circ)\\
&=\sum_{\theta}	\mu_{0}(\theta)u_S(\theta,a^\circ) = v_0.
\end{align*}
Therefore, Sender receives a strictly higher payoff from $\pi^{\epsilon}$ than $v_0$. This completes the proof.

The rest of the proof shows that the set $\Delta(\Theta)/\{\cup_{a\in A} P_a\}$ is negligible in $\Delta(\Theta)$.

Define $H_{a,a'} \equiv \{\mu\in\Delta(\Theta)|\sum_\theta\mu(\theta) (u_R(\theta,a)-u_R(\theta,a'))=0\}$ for any $a\neq a'$. Since by assumption, $u_R(\cdot,a)-u_R(\cdot,a')\neq {\bf {0}}$, which implies $J_{a,a'}\equiv \{\mu\in \mathbb{R}^{|\Theta|}|\sum_\theta\mu(\theta) (u_R(\theta,a)-u_R(\theta,a'))=0\}$ is a hyperplane in $\mathbb{R}^{|\Theta|}$. Notice that $H_{a,a'}=J_{a,a'}\cap \Delta (\Theta)$, which is the intersection of a hyperplane with a simplex. Since the hyperplane includes $\mathbf{0}$ and $\Delta (\Theta)$ doesn't, they have to either be parallel with no intersection, or their intersection is in a lower dimensional space, which has measure $0$ in $\Delta(\Theta)$.

For any $\mu\in \Delta(\Theta)/\{\cup_{a\in A} P_a\}$, since the maximizer of $\sum_\theta\mu(\theta)u_R(\theta,a)$ is not unique, there exists $a,a'$ such that $\sum_\theta\mu(\theta) (u_R(\theta,a)-u_R(\theta,a'))=0$. So $\Delta(\Theta)/\{\cup_{a\in A} P_a\}\subset \cup_{a\neq a'} H_{a,a'}$, which implies $\Delta(\Theta)/\{\cup_{a\in A} P_a\}$ is a negligible set in $\Delta(\Theta)$.

\end{proof}
\begin{proof}[Proof of statement 2]
  For any generic prior $\mu^\circ\in \cup_{a\in A} P_a$, either $\mu^\circ \notin P_{\underline{a}}$ or $\mu^\circ \notin P_{\overline{a}}$. We consider the case $\mu^\circ \notin P_{\overline{a}}$, and the other case can be shown symmetrically. Similar as the previous argument, for $\epsilon$ sufficiently small, consider the outcome distribution $\pi^{\epsilon}\in\Delta(\Theta\times A)$:
\begin{equation*}
	\pi^{\epsilon}(\theta,a) =
	\begin{cases}
		\mu_{0}(\theta) & \text{if } \theta \ne \overline{\theta}, a = a^{\circ} \\
		\mu_{0}(\overline{\theta}) - \epsilon & \text{if } (\theta,a)  = (\overline{\theta}, a^{\circ}) \\
		\epsilon & \text{if } (\theta,a)  = (\overline{\theta}, \overline{a}) \\
		0 & \text{otherwise }
	\end{cases}
\end{equation*}
As we have shown in the proof of statement 1, for $\epsilon$ sufficiently small, $\pi^\epsilon$ is stable, and gives the Sender higher payoff than $v_0$. Therefore, the sender benefits from credible persuasion.

\end{proof}
\begin{proof}[Proof of statement 3]

Consider any fully revealing outcome distribution $\pi^*\in \Delta(\Theta\times A)$ which gives a strictly higher payoff to the Sender than every no-information outcome. Let $A^*(\theta) \equiv \argmax_{a\in A} u_R(\theta, a )$ denote the Receiver's best response correspondence. By definition, for every $(\theta,a)\in \supp(\pi^*)$, $a\in A^*(\theta)$. If $\pi^*$ is comonotone, then then from \autoref{theorem-cyclical} and \autoref{lemma: comonotone}, $\pi^*$ is credible so the result follows. If $\pi^*$ is not comonotone, then there exists $(\theta,a)$, $(\theta',a')$ in the support of $\pi^*$ where $\theta>\theta'$ and $a<a'$. Let $d=\max\{\theta'-\theta|(\theta,a),(\theta',a')\in\supp(\pi^*),\theta>\theta',a<a'\}$ denote the largest distance of states between those ``non-monotone" pairs. Suppose $(\theta_1,a_1)$, $(\theta_2,a_2)$ is a pair that induces the largest distance, where $\theta_1<\theta_2$ and $a_1>a_2$. 
 
 Let $\epsilon=\min\{\pi^*(\theta_1,a_1),\pi^*(\theta_2,a_2)\}$, and construct the following outcome distribution $\pi'$:
 \begin{itemize}[nolistsep]
  	\item $\pi'(\theta_1,a_1)=\pi^*(\theta_1,a_1)-\varepsilon$, $\pi'(\theta_2,a_2)=\pi^*(\theta_2,a_2)-\varepsilon$
		\item $\pi'(\theta_1,a_2)=\pi^*(\theta_1,a_2)+\varepsilon$, $\pi'(\theta_2,a_1)=\pi^*(\theta_2,a_1)+\varepsilon$
		\item $\pi'(\theta,a)=\pi^*(\theta,a)$ for any other $(\theta,a)$
 \end{itemize}
 
 For any $a\notin\{a_1,a_2\}$, the obedient constraint under $\pi'$ is the same as under $\pi^*$, so the obedient constraint still holds. For $a\in\{a_1,a_2\}$, we show that the obedient constraint is still satisfied. 
 
 Since $u_R(\theta,a)$ is supermodular, by Lemma 2.8.1 of \cite{topkis2011}, $A^*(\theta)$ is increasing in $\theta$ in the induced set order. That is, for any $\theta>\theta'$, $a\in A^*(\theta)$, and $a'\in A^*(\theta')$, we have $\max\{a,a'\}\in A^*(\theta)$ and $\min\{a,a'\}\in A^*(\theta')$. Since $a_1\in A^*(\theta_1)$ and $a_2\in A^*(\theta_2)$, we have $a_1\in A^*(\theta_2)$ and $a_2\in A^*(\theta_1)$. Therefore, $\pi'$ also satisfies obedient. Moreover, the Sender's payoff from $\pi'$ is greater than from $\pi^*$, because $u_S$ is supermodular.
 
 Now we can iterate the process until $d=0$, and we construct an outcome distribution which is comonotone, obedient, and gives the Sender a higher payoff than $\pi^*$. Since the Sender's payoff from $\pi^*$ is strictly greater than any no-information outcome, the Sender benefits from credible persuasion.

\end{proof}

\subsection{Proof of \autoref{proposition: no cost}}
  From Theorem 1 of \citet{mensch2019monotone}, if both $u_S$ and $u_R$ are supermodular and $|A|=2$, there exists a KG optimal outcome distribution that is comonotone. Then by \autoref{theorem-cyclical} and \autoref{lemma: comonotone}, such an outcome distribution is stable. Moreover, if in addition $u_S$ is strictly supermodular, any KG optimal outcome distribution is comonotone. So any KG optimal outcome distribution is stable.

\subsection{Proof of \autoref{prop: adverse selection}}

For each buyers' belief over quality, $\mu \in \Delta(\Theta)$, let $\underline{\theta}_\mu$ denote the smallest $\theta$ in the support of $\mu$. In addition, for each  $\mu \in \Delta(\Theta)$, let $\phi_\mu(x) \equiv E_\mu[v(\theta)|\theta\leq x]$ denote the corresponding expected value to buyers when the quality threshold is $\theta\le x$.\footnote{For $x$ less than $\underline{\theta}_\mu$ we set $\phi_\mu(x)=v(\underline{\theta}_\mu)$.}  Clearly, $\phi_\mu(\cdot)$ is increasing and $\phi_\mu(1)=E_\mu[v(\theta)]$.

\begin{lemma}\label{lemma:adverse:fixedpoint}
For every $\mu \in \Delta(\Theta)$, there exists a largest fixed point $\theta^*_\mu\in(\underline{\theta}_\mu,1)$ such that  $\phi_\mu(\theta^*_\mu)=\theta^*_\mu$. Moreover, for any $\theta\in(\theta^*_\mu,1]$, $\phi_\mu(\theta)<\theta$.
\end{lemma}
\begin{proof}
	Since $\phi_\mu(\underline{\theta}_\mu)=v(\underline{\theta}_\mu)>\underline{\theta}_\mu$, $\phi_\mu(1)=E_\mu[v(\theta)]<1$, and $\phi_\mu(\cdot)$ is increasing, from Tarski's fixed point theorem, there exists a largest fixed point $\theta^*_\mu\in (\underline{\theta}_\mu,1)$ such that $\phi_\mu(\theta^*_\mu)=\theta^*_\mu$. To see the second statement, suppose there exists $\theta\in(\theta^*_\mu,1)$ such that $\phi_\mu(\theta)\geq \theta$, again from Tarski's fixed point theorem, there exists a fixed point $\theta'\in (\theta^*_\mu,1)$, which contradicts to $\theta_{\mu}^*$ being the largest fixed point.
\end{proof}

\begin{lemma}\label{lemma: BNEexist}
Let $\lambda\in \Delta(\Theta \times M)$ be a test, and for every $m\in M$ let $\mu_m \in \Delta(\Theta)$ denote the buyers' posterior belief after observing message $m$.  The following strategy profile is a BNE in the game $\langle G,\lambda \rangle$: $\alpha_S(\theta,m)=\theta$, $\beta_1(m)=\beta_2(m)=\theta^*_{\mu_m}$.
\end{lemma}
\begin{proof}
For every message $m$, since $\phi_{\mu_m}(\theta_{\mu_{m}}^*)=\theta_{\mu_{m}}^*$, each buyer's expected payoff is 0. Any deviation to a lower bid also gives a payoff of zero. From \autoref{lemma:adverse:fixedpoint}, for any $\theta\in(\theta^*_{\mu_m},1]$, $\phi_{\mu_m}(\theta)<\theta$, so any deviation to a bid higher than $\theta^*_{\mu_m}$ would lead to a negative payoff. Therefore no buyer has incentive to deviate.
\end{proof}

\begin{lemma}\label{lemma:adverse:IC}
Let $(\lambda^*, \sigma^*)$ a WD-IC profile. For each message $m$, let $p(m) \equiv \max\{\beta_1^*(m),\beta_2^*(m)\}$ denote the equilibrium market price  in the game $\langle G,\lambda^* \rangle$. Then $\phi_{\mu_m}(p(m))=p(m) \text{ for each } m\in M$.
\end{lemma}
\begin{proof}
Suppose $\phi_{\mu_m}(p(m))<p(m)$, then the winning buyer's payoff is negative, and can profitably deviate to bid 0. Now suppose $\phi_{\mu_m}(p(m))>p(m)$, we show that at least one buyer has an incentive to bid a higher price. 

If $\beta^*_1(m)\neq \beta^*_2(m)$, then the losing bidder can profitably deviate. Since $\phi_{\mu_m}(\cdot)$ is increasing, there exists small enough $\varepsilon$ such that $\phi_{\mu_m}(p(m)+\varepsilon)>p(m)+\varepsilon$. So the losing bidder can deviate to bidding $p(m)+\varepsilon$ and receives a strictly positive payoff.

If $\beta^*_1(m)= \beta^*_2(m)=b$ for some $b$, we show that both buyers have an incentive to deviate. Let $K\equiv \phi_{\mu_m}(b)-b>0$. Since ties are broken evenly, each buyer's payoff is $\frac{1}{2}P_{\mu_m}(\theta\leq b)K$.
By letting $\varepsilon<\frac{K}{2}$, we have 
\begin{align*}
\phi_{\mu_m}(b+\varepsilon)-b-\varepsilon \geq \phi_{\mu_m}(b)-b-\varepsilon = K-\varepsilon > \frac{K}{2}.
\end{align*}
So if either of the bidders deviates to bidding $b+\varepsilon$, he receives a payoff of $P_{\mu_m}(\theta\leq b+\varepsilon)[\phi_{\mu_m}(b+\varepsilon)-b-\varepsilon]> \frac{1}{2}P_{\mu_m}(\theta\leq b)K$, which is profitable.
\end{proof}

\begin{lemma}\label{lemma:adverse:cyclical}
If a profile $(\lambda^*,\sigma^*)$ is credible and WD-IC, then there exists a set $E \subset \Theta\times M$ such that $\lambda^*(E)=1$, and for any $(\theta,m)$, $(\theta',m')\in E$ with $\theta>\theta'$ and $p(m)>p(m')$,
\[
\max\{\theta,p(m)\}+\max\{\theta',p(m')\}\geq \max \{ \theta, p({m'}) \} + \max \{\theta',p(m)\}.
\]
\end{lemma}
\begin{proof}
Since $(\lambda^*, \sigma^*)$ is WD-IC, trade only happens when the seller's ask price $\alpha^*(\theta, m) =\theta$ is higher than the prevailing market price $p(m) = \max\{\beta^*_1(m),\beta^*_2(m)\} $. The seller's payoff function can therefore be simplified as
\[
u_S(\theta,\sigma^*(\theta,m))=u_S(\theta,\alpha^*(\theta,m),\beta^*_1(m),\beta^*_2(m))=\max\{\theta,p(m)\}.
\]
Recall that credibility requires
\[
\lambda\in \argmax_{\lambda'\in \Lambda(\mu_0, \lambda_M)}\int u_S(\theta,\sigma^*(\theta,m))  \, d\lambda'(\theta,m).
\]
Let $v_S(\theta,m)\equiv u_S(\theta,\sigma^*(\theta,m))=\max\{\theta,p(m)\}$. From Theorem 1 of \cite{optimaltransport2009}, $\lambda$ is $v_S$-cyclically monotone. That is, there exists a set  $E \subset \Theta\times M$ such that $\lambda^*(E)=1$, and for any sequence $\left(\theta_k,m_k\right)_{k=1}^n\in E$,
\[\sum_{k=1}^n v_S(\theta_k,m_k)\geq \sum_{k=1}^n v_S(\theta_k,m_{k+1}).\]
	
Suppose $(\theta,m), (\theta',m')\in E$, $\theta>\theta'$, and $p(m)>p({m'})$. Then $v_S$-cyclical monotonicity implies that 
\[
v_S(\theta,m)+v_S(\theta',m')\geq v_S(\theta,m')+v_S(\theta',m),
\] 
which is
\[
\max\{\theta,p(m)\}+\max\{\theta',p(m')\}\geq \max \{ \theta, p({m'}) \} + \max \{\theta',p(m)\}.
\]
\end{proof}

In light of \cref{lemma:adverse:cyclical}, for every credible profile $(\lambda^*,\sigma^*)$ we will focus only on pairs $(\theta,m)\in E$. We will  use $\proj(E)\equiv \{m \in M: (\theta,m) \in E\}$ to denote the projection of $E$ onto the message space.
\smallskip

Let $\underline{p}= \inf \{p(m)| m\in \proj(E)\}$ be the lowest trading price across all messages. For each message $m$, let $\Theta(m)=\{\theta:(\theta,m)\in E\}$ be the set of $\theta$ that is matched with $m$.
\begin{lemma}\label{lemma:adverse:nointersection}
Let $(\lambda^*, \sigma^*)$ be a credible and WD-IC profile. For every message  $\hat{m}\in \proj(E)$ such that $\hat{p}\equiv p(\hat{m})= \max\{\beta^*_1(\hat{m}),\beta^*_2(\hat{m})\}>\underline{p}$, we have $\Theta(\hat{m})\cap (\underline{p},\infty)=\emptyset$.
\end{lemma}
\begin{proof}

To prove the lemma, suppose by contradiction that there exists $\hat{\theta}\in \Theta(\hat{m})\cap (\underline{p},\infty)$. By the definition of $\underline{p}$, there exists $p'$ with $\underline{p} < p' <\hat{\theta}$ such that $p'=p(m')$ for some $m' \in \proj(E)$. Since in equilibrium $p' = E_{\mu_{m'}}[v(\theta)|\theta\in \Theta(m')\cap [0, p']]$, there also exists $\theta' \in \Theta(m')$ such that $\theta' < p'$. Since $({\theta}',{m}'), (\hat{\theta},\hat{m})\in E$, by \cref{lemma:adverse:cyclical}, we have
\begin{align*}
\max\{\theta',\hat{p}\}+\max\{\hat{\theta}, p'\}\leq \max\{{\theta}', p'\}+\max\{\hat{\theta},\hat{p}\} 
\end{align*}
Since $\theta'<p'$ by construction, we have
\begin{equation} \label{equation: adverse1}
\max\{\theta',\hat{p}\}+\max\{\hat{\theta}, p'\} \le {p}' + \max\{\hat{\theta},\hat{p}\}
\end{equation}
Note also that
\begin{equation} \label{equation: adverse2}
    {p}' + \max\{\hat{\theta},\hat{p}\} < \hat{p}+\hat{\theta}.
\end{equation}
The inequality above follows by considering two possibilities for  $\max\{\hat{\theta},\hat{p}\}$: either $\hat{\theta}\ge \hat p$, in which case ${p}' + \max\{\hat\theta, \hat p\} = {p}' + \hat\theta <\hat{p} + \hat\theta$; or $\hat\theta < \hat p$, in which case ${p}' + \max\{\hat\theta,\hat p\} = {p}' + \hat p < \hat{p} + \hat\theta$ as well.

Combining \eqref{equation: adverse1} and \eqref{equation: adverse2}, and noticing $\theta' <p'< \hat p$ and $p' < \hat\theta$ yield
\begin{equation*}
\max\{\theta',\hat{p}\}+\max\{\hat{\theta}, p'\} < \hat p+\hat\theta= \max\{\theta', \hat p\}+\max\{\hat\theta, p'\}
\end{equation*}
which is a contradiction.

\end{proof}

\begin{proof}[Proof of the \cref{prop: adverse selection}]

To prove our result, we first calculate the seller's profit from an arbitrary credible and WD-IC profile $(\lambda^*,\sigma^*)$. We then show that there exists another credible and WD-IC profile $(\lambda^0,\sigma^0)$, where $\lambda^0$ is a null information structure, that leads to weakly higher profit for the seller.

Recall that seller's payoff function can be written as $u_S(\theta,\sigma^*(\theta,m))=\max\{\theta,p(m)\}$, so her ex-ante profit is
\begin{align*}
 & \int_{\Theta\times M} \max\{\theta,p(m)\}d\lambda^*(\theta,m) = \int_M \int_0^1 \, \max\{\theta,p(m)\} \, d\lambda^*(\theta|m) d \lambda^*_M (m)\\
=& \int_M \left[\int_0^{p(m)} \, p(m) \, d\lambda^*(\theta|m) +\int_{p(m)}^1 \theta \, d\lambda^*(\theta|m)\right]d \lambda^*_M (m) \\
=& \int_M \left[ p(m)P_{\lambda^*(\theta|m)}(\theta\leq p(m)) +\int_{p(m)}^1 \theta \, d\lambda^*(\theta|m)\right]d \lambda^*_M (m) 
\end{align*}
By \cref{lemma:adverse:IC}, $p(m)=E_{\lambda^*(\theta|m)}[v(\theta)|\theta\leq p(m)]$, so we can write the integral above as
\begin{align*}
 &\int_M \left[ E_{\lambda^*(\theta|m)}[v(\theta)|\theta\leq p(m)]P_{\lambda^*(\theta|m)}(\theta\leq p(m)) +\int_{p(m)}^1 \theta \, d\lambda^*(\theta|m)\right]d \lambda^*_M (m)\\
 = & \int_M \left[ \int_0^{p(m)} v(\theta)d \lambda^*(\theta|m)+\int_{p(m)}^1 \theta \, d\lambda^*(\theta|m)\right]d \lambda^*_M (m) 
\end{align*}
By \cref{lemma:adverse:nointersection}, for every $m\in \proj(E)$, if $p(m)>\underline{p}$ then $\Theta(m)\cap (\underline{p},\infty)=\emptyset$, so the seller's profit from $(\lambda^*,\sigma^*)$ can be further simplified to
\begin{align}
& \int_M \left[ \int_0^{\underline{p}} v(\theta)d \lambda^*(\theta|m)+\int_{\underline{p}}^1 \theta \, d\lambda^*(\theta|m)\right]d \lambda^*_M (m) \nonumber \\
=&\int_0^{\underline{p}} v(\theta)d \mu_0 (\theta) + \int_{\underline{p}}^1 \theta \, d\mu_0(\theta). \label{equation: seller profit*}
\end{align}

Having calculated the seller's profit from $(\lambda^*,\sigma^*)$, next we will construct another credible and WD-IC profile $(\lambda^0,\sigma^0)$ with a weakly higher profit, where $\lambda_0$ is the null information structure $\mu_0\times \delta_{m_0}$.

From  \cref{lemma:adverse:nointersection}, for every $m\in \proj(E)$,
\[
\phi_{\mu_m}(p(m)) = E_{\mu_m}[v(\theta)|\theta\leq p(m)] = E_{\mu_m} [v(\theta) | \theta\leq \underline{p}] = \phi_{\mu_m}(\underline{p});
\]
in addition, from \cref{lemma:adverse:IC}, $\phi_{\mu_m}(p(m))=p(m)$ for every message $m \in \proj(E)$. Combining these yields
\[
\phi_{\mu_m}(\underline{p})=\phi_{\mu_m}(p(m))=p(m)\geq \underline{p}.
\]
Taking expectation over all messages, we have $\phi_{\mu_0}(\underline{p})\geq \underline{p}$. By Tarski's fixed point theorem, there exists a largest $p^0\in[\underline{p},1)$ such that the $\phi_{\mu_0}(p^0)=p^0$.

Using a similar argument as that in \cref{lemma: BNEexist}, the strategy profile $\sigma^0$ where the seller plays her weakly dominant strategy $\alpha^0(\theta,m_0)=\theta$, and buyers play $\beta_1^0(m_0) = \beta_2^0(m_0) = p^0$ is a BNE in the game $\langle G, \lambda^0 \rangle$.

It remains to show that the seller's profit from $(\lambda^0,\sigma^0)$ is weakly higher than that from $(\lambda^*,\sigma^*)$. Under $(\lambda^0,\sigma^0)$ the seller's profit is
\begin{align}
\int_0^1 \max\{\theta,p^0\} d\mu_0(\theta) &=\int_0^{p^0} p^0 d\mu_0 (\theta) + \int_{p^0}^1 \theta d\mu_0(\theta) \nonumber \\
&= p^0 P_{\mu_0}(\theta\leq p^0) +\int_{p^0}^1 \theta d\mu_0(\theta)\nonumber\\
&=E_{\mu_0}[v(\theta)|\theta\leq p^0]P_{\mu_0}(\theta\leq p^0)+\int_{p^0}^1 \theta d\mu_0(\theta)\nonumber\\
&=\int_0^{p^0} v(\theta) d\mu_0(\theta)+\int_{p^0}^1 \theta d\mu_0(\theta) \label{equation: seller profit0}
\end{align}
Comparing \eqref{equation: seller profit*} and \eqref{equation: seller profit0}, since  $p^0\geq \underline{p}$ and $v(\theta)>\theta$ for all $\theta$, it follows that 
\[
\int_0^{p^0} v(\theta) d\mu_0+\int_{p^0}^1 \theta d\mu_0\geq \int_0^{\underline{p}} v(\theta)d \mu_0+\int_{\underline{p}}^1 \theta  d\mu_0.
\]
The seller's profit under $(\lambda^0,\sigma^0)$ is therefore weakly higher than that from $(\lambda^*,\sigma^*)$.

\end{proof}

\subsection{Proof of \autoref{proposition: finite} \label{section: proof finite}}
The following lemmas will be useful in our proofs.

\begin{lemma} \label{lemma: continuous solution} The following correspondence 
\[B(\mu,\nu)\equiv \arg\max_{\lambda\in\Lambda (\mu,\nu)}\sum_{\theta,m}\lambda(\theta,m)u_S(\theta,\sigma(m))\]
is upper hemi-continuous with respect to $(\mu,\nu)$. Thus, the value function \[V(\mu,\nu)\equiv\max_{\lambda\in\Lambda (\mu,\nu)}\sum_{\theta,m}\lambda(\theta,m)u_S(\theta,\sigma(m))\] is continuous.
\end{lemma}
\begin{proof}
The first statement follows directly from Theorem 1.50 of \citet{santambrogio2015optimal}. For any sequence $(\lambda_k,\mu_k,\nu_k)\rightarrow (\lambda,\mu,\nu)$ so that $\lambda_k\in B(\mu_k,\nu_k)$ for all $k$, we have $\lambda\in B(\mu,\nu)$. Then $V(\mu,\nu)=\sum_{\theta,m}\lambda(\theta,m)u_S(\theta,\sigma(m))=\lim_{k\rightarrow \infty}\sum_{\theta,m}\lambda_k(\theta,m)u_S(\theta,\sigma(m))=\lim_{k\rightarrow \infty} V(\nu_k,\nu_k)$, which proves the second statement.
\end{proof}

\begin{lemma} \label{lemma: integer extreme points}
Suppose $\mu\in \mathcal{F}^N_\Theta$ and $\nu\in \mathcal{F}^N_M$, then the extreme points of $\Lambda(\mu,\nu)$ is contained in $X^N(\mu,\nu)$.
\end{lemma}
\begin{proof}
Consider the set $Y^N(\mu,\nu)=\{f\in \mathbb{R}_+^{|\Theta|\times |M|}: \sum_\theta f(\theta,\cdot)=N \nu(\cdot),\sum_m f(\cdot,m)=N \mu(\cdot)\}$. From Corollary 8.1.4 of \citet{brualdi2006combinatorial}, the extreme points of $Y^N(\mu,\nu)$ is contained in $Z^N(\mu,\nu)=\{f\in \mathbb{N}^{|\Theta|\times |M|}: \sum_\theta f(\theta,\cdot)=N \nu(\cdot),\sum_m f(\cdot,m)=N \mu(\cdot)\}$. Since $\Lambda(\mu,\nu)=\{\frac{f}{N}: f\in Y^N(\mu,\nu)\}$ and $X^N(\mu,\nu)=\{\frac{f}{N}: f\in Z^N(\mu,\nu)\}$, the extreme pints of $\Lambda(\mu,\nu)$ is contained in $X^N(\mu,\nu)$.
\end{proof}

\begin{lemma} \label{lemma: uhc + unique = cont}
Let $X,Y$ be metric spaces and $\Gamma:X\rightrightarrows Y$ be a correspondence. If $\Gamma$ is upper hemi-continuous at $x_0\in X$, and $\Gamma(x_0)=\{y_0\}$ for some $y_0\in Y$, then $\Gamma$ is continuous at $x_0$.
\end{lemma}
\begin{proof}
For any $\epsilon>0$, let $B(y_0,\epsilon) \subseteq Y$ denote the $\epsilon$-ball centered at $y_0$. We will show that there exists $\delta>0$ such that for all $|x-x_0|<\delta$, $\Gamma(x)\cap B(y_0,\epsilon) \ne \emptyset$, which implies that $\Gamma$ is lhc and therefore continuous.

Now since $\Gamma(x) = \{y_0\} \subseteq B(y_0,\epsilon)$ and $\Gamma$ is uhc at $x_0$, it follows that there exists $\delta>0$ such that $\Gamma(x)\subseteq B(y_0,\epsilon)$ for all $|x-x_0|<\delta$, so $\Gamma(x)\cap B(y_0,\epsilon) \ne \emptyset$ for all $|x-x_0|<\delta$, which completes the proof.
\end{proof}

\begin{proof}[Proof of \autoref{proposition: finite} statement 1]
First suppose $(\lambda^*,\sigma^*)$ is not credible. Then there exists $\lambda'\in \Lambda(\mu_0,\lambda^*_{M})$ (recall $\mu_0$ is the prior distribution on $\Theta$) and $\epsilon_0>0$ such that
\[
\sum_{\theta,m} \lambda^*(\theta,m)u_S(\theta,\sigma^*(m)) < \sum_{\theta,m} \lambda'(\theta,m) u_S(\theta,\sigma^*(m)) -  \epsilon_0  
\]
By continuity, there exists $\epsilon_1>0$ such that for all  $|{\lambda} - \lambda^* |<\epsilon_1$  we have
\begin{equation} \label{equation: bound1}
 \sum_{\theta,m} {\lambda}(\theta,m)u_S(\theta,\sigma^*(m)) < \sum_{\theta,m} {\lambda}'(\theta,m) u_S(\theta,\sigma^*(m)) - \epsilon_0/2
\end{equation}
By \cref{lemma: continuous solution}, there exists $\varepsilon_2>0$ such that for all $|\mu - \mu_0| < \varepsilon_2$ and $|\nu - \lambda^*_M| <\epsilon_2 $, there exists ${\lambda}\in \Lambda(\mu,\nu)$ with
\begin{equation} \label{equation: bound 2}
    \sum_{\theta,m} {\lambda} (\theta,m) u_S(\theta,\sigma^*(m)) >\sum_{\theta,m} {\lambda}'(\theta,m) u_S(\theta,\sigma^*(m)) -\epsilon_0/2.
\end{equation}

Moreover, since the receiver is choosing only pure strategies, there exists $\varepsilon_3$ such that for any $\sigma$ where $|\sigma-\sigma^*|<\varepsilon_3$, $\sigma=\sigma^*$.

Now let $\epsilon=\min\{\varepsilon_1,\frac{\varepsilon_2}{|\Theta|\times |M|},\varepsilon_3\}$,
By assumption, there exists a finite-sample, credible and R-IC profile $(\nu^N, \phi^N, \sigma^N)$ such that $|\nu^N-\lambda^*_M| < \varepsilon$, $|\sigma^N-\sigma|<\varepsilon$ and $P(|\phi^N(F_{\Theta}^N) - \lambda | < \varepsilon) > 1-\varepsilon$.

Under such a finite-sample profile, $\sigma^N=\sigma^*$, and there exists $F^N_{\Theta} \in \mathcal{F}^N_{\Theta}$, realized with positive probability, such that $\tilde{\lambda}^* = \phi^N(F^N_{\Theta} ) $ satisfies $|\tilde{\lambda}^* - \lambda^*| < \min \{\epsilon_1, \frac{\epsilon_2}{|\Theta|\times|M|} \}$.

 Now since $|\tilde{\lambda}^* - \lambda^*| <\epsilon_1$, by \eqref{equation: bound1} we know that
\begin{equation} \label{equation: bound 3}
 \sum_{\theta,m} \tilde{\lambda}^*(\theta,m)u_S(\theta,\sigma^*(m)) < \sum_{\theta,m} {\lambda}'(\theta,m) u_S(\theta,\sigma^*(m)) - \epsilon_0/2    
\end{equation}
In addition, since $|\tilde{\lambda}^* - \lambda^*| <  \frac{\epsilon_2}{|\Theta|\times|M|}$, we know $F^N_{\Theta} = \tilde{\lambda}^*_{\Theta}$ satisfies $| F^N_{\Theta} - \mu_0| <\epsilon_2$, and $\nu^N = \tilde{\lambda}^*_{M} $ satisfies $|\nu^N - \lambda^*_M| <\epsilon_2$, so by \eqref{equation: bound 2} there exists $\tilde{\lambda}' \in \Lambda(F^N_{\Theta}, \nu^N)$ such that
\begin{equation}  \label{equation: bound 4}
    \sum_{\theta,m} \tilde{\lambda}' (\theta,m) u_S(\theta,\sigma^*(m)) >\sum_{\theta,m} {\lambda}'(\theta,m) u_S(\theta,\sigma^*(m)) -\epsilon_0/2
\end{equation}
Combining \eqref{equation: bound 3} and \eqref{equation: bound 4}, we have
\begin{equation*}
   \sum_{\theta,m} \tilde{\lambda}' (\theta,m) u_S(\theta,\sigma^*(m)) > \sum_{\theta,m} \tilde{\lambda}^*(\theta,m)u_S(\theta,\sigma^*(m)),
\end{equation*}
Note that $\tilde{\lambda}'\in \Lambda(F^N_{\Theta},\nu^N)$, but by \cref{lemma: integer extreme points}, we can replace $\tilde{\lambda}'$  with an extreme point in $X^N(F^N_{\Theta},\nu^N)$, and the above inequality still holds. That is, there exists $\hat{\lambda}'\in X^N(F^N_{\Theta},\nu^N)$ such that
\[\sum_{\theta,m} \hat{\lambda}'(\theta,m) u_S(\theta,\sigma^*(m))>\sum_{\theta,m} \tilde{\lambda}^* (\theta,m) u_S(\theta,\sigma^*(m)),\]
Notice that $\tilde{\lambda}^*$ and $\hat{\lambda}'$ are both in $X^N(F_\Theta^N,\nu^N)$, which is a contradiction since by the credibility of $(\nu^N, \phi^N,\sigma^N)$
\begin{equation*}
   \tilde{\lambda}^* = \phi^N(F^N_{\Theta}) = \argmax_{\lambda\in X^N(F^N_{\Theta},\nu^N )} \sum_{\theta,m} \lambda(\theta,m) u_S(\theta,\sigma^*(m)).
\end{equation*}

Second, suppose $(\lambda^*,\sigma^*)$ violates R-IC. Then there exists $\sigma'$ such that
\[
\sum_{\theta,m} \lambda^*({\theta,m}) u_R(\theta,\sigma'(m)) > \sum_{\theta,m} \lambda^*({\theta,m}) u_R(\theta,\sigma^*(m))
\]
By continuity, there exist $\eta>0$ and $\varepsilon_4>0$ such that for all $\lambda'$ satisfying $|\lambda^*-\lambda'|<\varepsilon_4$, we have
\[
\sum_{\theta,m} \lambda'(\theta,m) u_R(\theta,\sigma'(m)) - \sum_{\theta,m} \lambda'({\theta,m}) u_R(\theta,\sigma^*(m))\geq \eta>0
\]

Let $d\equiv\max_{\theta,a} u_R(\theta,a) - \min_{\theta,a} u_R(\theta,a)$ denote the gap between the Receiver's highest and lowest payoffs. Let $\varepsilon_5 \leq \frac{\eta}{d+\eta}$ and $\varepsilon=\min\{\varepsilon_3,\varepsilon_4,\varepsilon_5\}$. By assumption, there exists a credible and R-IC finite-sample profile $(\nu^N, \phi^N,\sigma^N)$ such that $Pr(|\phi^N(F^N_{\Theta}) - \lambda^*| \leq \varepsilon) > 1-\epsilon$, and $\sigma^N = \sigma^*$. We will show that in the finite sample profile $(\nu^N, \phi^N,\sigma^N)$, the Receiver can profitably deviate from $\sigma^N=\sigma^*$ to $\sigma'$, which contradicts $(\nu^N, \phi^N,\sigma^N)$ being R-IC.

By choosing $\sigma^*$ the Receiver obtains payoff
\begin{equation*}
\sum_{F^N_{\Theta}\in \mathcal{F}^N_{\Theta} }  P^{N}(F^N_{\Theta}) \sum_{\theta, m} \phi^N(\theta,m| F^N_{\Theta} ) u_{R} (\theta, \sigma^* ( m))    
\end{equation*}
By contrast, the Receiver obtains
\begin{equation*}
\sum_{F^N_{\Theta}\in \mathcal{F}^N_{\Theta} }  P^{N}(F^N_{\Theta}) \sum_{\theta, m} \phi^N(\theta,m| F^N_{\Theta} ) u_{R} (\theta, \sigma'(m))    
\end{equation*}
from choosing $\sigma'$. Denote $E^N \equiv \{F^N_{\Theta}: |\phi^N(F^N_{\Theta}) - \lambda^*|\leq \delta\}$ so $Pr(E^N) > 1-\epsilon$. By switching from $\sigma^*$ to $\sigma'$, the Receiver obtains an extra payoff of
\begin{align*}
    & \sum_{F^N_{\Theta}\in \mathcal{F}^N_{\Theta} }  P^{N}(F^N_{\Theta})  \left[ \sum_{\theta, m} \phi^N(\theta,m| F^N_{\Theta} ) u_{R} (\theta, \sigma^*(m)) -  \sum_{\theta, m} \phi^N(\theta,m| F^N_{\Theta} ) u_{R} (\theta, \sigma'(m)) \right]\\
    = & \sum_{F^N_{\Theta} \in E^N }  P^{N}(F^N_{\Theta})  \left[ \sum_{\theta, m} \phi^N(\theta,m| F^N_{\Theta} ) u_{R} (\theta, \sigma^*(m)) -  \sum_{\theta, m} \phi^N(\theta,m| F^N_{\Theta} ) u_{R} (\theta, \sigma'(m)) \right]\\
    & + \sum_{F^N_{\Theta}\notin E^N }  P^{N}(F^N_{\Theta})  \left[ \sum_{\theta, m} \phi^N(\theta,m| F^N_{\Theta} ) u_{R} (\theta, \sigma^*(m)) -  \sum_{\theta, m} \phi^N(\theta,m| F^N_{\Theta} ) u_{R} (\theta, \sigma'(m)) \right]
\end{align*}
Note that $\sum_{\theta, m} \phi^N(\theta,m| F^N_{\Theta} ) u_{R} (\theta, \sigma^*(m)) -  \sum_{\theta, m} \phi^N(\theta,m| F^N_{\Theta} ) u_{R} (\theta, \sigma'(m)) \ge \eta$ for all $F^N_{\Theta}\in E^N$, while for all $F^N_{\Theta}\notin E^N$, $\sum_{\theta, m} \phi^N(\theta,m| F^N_{\Theta} ) u_{R} (\theta, \sigma^*(m)) -  \sum_{\theta, m} \phi^N(\theta,m| F^N_{\Theta} ) u_{R} (\theta, \sigma'(m)) \ge -d$. 
Together they imply,
\begin{align*}
& \sum_{F^N_{\Theta}\in \mathcal{F}^N_{\Theta} }  P^{N}(F^N_{\Theta})  \left[ \sum_{\theta, m} \phi^N(\theta,m| F^N_{\Theta} ) u_{R} (\theta, \sigma^*(m)) -  \sum_{\theta, m} \phi^N(\theta,m| F^N_{\Theta} ) u_{R} (\theta, \sigma'(m)) \right]\\
\ge & \eta P^N(E^N) - d (1- P^N(E^N)).
\end{align*}
Since $P^N(E^N) > 1 - \epsilon$, we have
\begin{align*}
& \sum_{F^N_{\Theta}\in \mathcal{F}^N_{\Theta} }  P^{N}(F^N_{\Theta})  \left[ \sum_{\theta, m} \phi^N(\theta,m| F^N_{\Theta} ) u_{R} (\theta, \sigma^*(m)) -  \sum_{\theta, m} \phi^N(\theta,m| F^N_{\Theta} ) u_{R} (\theta, \sigma'(m)) \right]\\
> & (1-\epsilon)\eta - \epsilon d = \eta -\epsilon( \eta + d) \ge 0
\end{align*}
This contradicts the R-IC of $(\nu^N,\phi^N,\sigma^N)$.

\end{proof}

\begin{proof}[Proof of \autoref{proposition: finite} statement 2]

For each $N \ge 1$, define $\sigma^N=\sigma^*$, $\nu^N\in \arg\min_{\nu\in \mathcal{F}^N_M} |\lambda^*_M-\nu|$, and $\phi^N:\mathcal{F}_\Theta^N\rightarrow \cup _{F_\Theta^N\in \mathcal{F}_\Theta^N} X(F^N_\Theta,\nu^N)$ by
\[
\phi(F^N_\Theta)\in \argmax_{\lambda \in X^N(F^N_\Theta,\nu^N)} \sum_{\theta,m} \lambda(\theta,m) u_S(\theta,\sigma^*(m)).
\]
By construction, for every $N$, $(\nu^N, \phi^N, \sigma^N)$ is credible and $|\sigma^N-\sigma^*|=0$. It remains to show that for every $\varepsilon>0$, there exists large enough $N$, such that 
\begin{enumerate}
    \item $|\nu^N - \lambda^*_M| <\varepsilon$;
    \item  $P(|\phi^N(F_{\Theta}^N)-\lambda^*| < \varepsilon)>1-\varepsilon$;
    \item $(\nu^N, \phi^N, \sigma^N)$ is R-IC.
\end{enumerate}

From the denseness of rational numbers, we know that $\nu^N\rightarrow \lambda^*_M$ as $N\rightarrow \infty$ so the first statement follows.

To prove the second statement, note that since $(\lambda^*,\sigma^*)$ is strictly credible, $\lambda^*$ is the unique maximizer to 
\[\max _{\lambda\in \Lambda(\mu,\nu)} \sum_{\theta,m}\lambda(\theta,m)u_S(\theta,\sigma^*(m)).\]
From \cref{lemma: continuous solution}, the best response correspondence $B(\mu,\nu)\equiv\arg\max _{\lambda\in \Lambda(\mu,\nu)} \sum_{\theta,m}\lambda(\theta,m)u_S(\theta_i,\sigma^*(m_j))$ is upper hemi-continuous. Since $B(\mu,\nu)=\{\lambda^*\}$ is a singleton, from \cref{lemma: uhc + unique = cont}, $B$ is continuous at $(\mu,\nu)$. Therefore, there exists $\delta>0$, so that for any $(\mu',\nu')$ such $|\mu-\mu'|<\delta$ and $|\nu-\nu'|<\delta$, we have $|\lambda'-\lambda^*|<\varepsilon$ for every $\lambda'\in B(\mu',\nu')$.

From the Glivenko–Cantelli theorem, for large $N$, $P(|F_\Theta^N-\mu_0|<\delta)>1-\varepsilon$. Pick $N$ large enough so that $P(|F_\Theta^N-\mu_0|<\delta)>1-\varepsilon$ and $|\nu^N-\mu_M^*|<\delta$.
Follows from the definition of $\phi$ and \cref{lemma: integer extreme points}, $\phi(F^N_\Theta)\in \argmax_{\lambda \in X^N(F^N_\Theta,\nu^N)} \sum_{\theta,m} \lambda(\theta,m) u_S(\theta,\sigma^*(m))\subset \argmax_{\lambda \in \Lambda(F^N_\Theta,\nu^N)} \sum_{\theta,m} \lambda(\theta,m) u_S(\theta,\sigma^*(m)$. So with at least $1-\varepsilon$ probability, $|\phi(F_{\Theta}^N)-\lambda^*|<\varepsilon$.

Lastly, we show that $(\nu^N, \phi^N, \sigma^N)$ is R-IC for large $N$. Since $(\lambda^*,\sigma^*)$ is strictly R-IC, for any $\sigma\neq \sigma^*$,
\[\sum_{\theta,m} 
\lambda^*(\theta,m) u_R(\theta,\sigma^*(m))>\sum_{\theta,m} \lambda^*(\theta,m) u_R(\theta,\sigma(m)).\] From continuity, there exists $\eta>0$ such that for any $\lambda$ such that $|\lambda^*-\lambda|<\varepsilon$, \[\sum_{\theta,m} 
\lambda(\theta,m) u_R(\theta,\sigma^*(m))-\sum_{\theta,m} \lambda(\theta,m) u_R(\theta,\sigma(m))\geq \eta> 0.\] As we have shown, for any $\varepsilon>0$, for large enough $N$, $Pr(|\phi(F^N_\Theta)-\lambda^*|\leq \varepsilon)\geq 1-\varepsilon$. Pick $\varepsilon\leq \frac{\eta}{d+\eta}$, then follow from the same argument above, we have
\[ \sum_{F^N_\Theta\in \mathcal{F}_{\Theta}^N} P^N(F^N_\Theta)\sum_{\theta,m} \phi(\theta,m|F^N_\Theta) u_R(\theta,\sigma^*(m))>\sum_{F^N_\Theta\in \mathcal{F}_{\Theta}^N} P^N(F^N_\Theta)\sum_{\theta,m} \phi(\theta,m|F^N_\Theta) u_R(\theta,\sigma(m))\]
for any $\sigma\neq \sigma^*$. So $(\nu^N,\phi^N,\sigma^N)$ is R-IC.

\end{proof}

\subsection{Proof of \autoref{proposition: S-R extensive form}}

\begin{proof}
The ``only if" direction:
Suppose $(\lambda,\sigma)$ is a pure-strategy SPNE outcome of the extensive-form game induced by a pure-strategy SPNE $(\lambda,\rho)$. Then under $(\lambda,\rho)$, the Sender should  have no profitable deviation to any other $\lambda' \in \Lambda$ and in particular, any $\lambda'$ that maintains the same marginal distribution on $M$. So for every $\lambda'\in\Lambda$ such that $\lambda'_M=\lambda_M$, we have
\begin{equation*}
\sum_\theta \sum_m \lambda(\theta,m)u_S(\theta,\rho(\lambda_M,m))\geq \sum_\theta \sum_m \lambda'(\theta,m)u_S(\theta,\rho(\lambda_M',m)).
\end{equation*}
Since $\rho(\lambda_M,m)=\rho(\lambda_M',m)=\sigma(m)$ for all $m\in M$, the equation above becomes
\begin{equation*}
\sum_\theta \sum_m \lambda(\theta,m)u_S(\theta,\sigma(m))\geq \sum_\theta \sum_m \lambda'(\theta,m)u_S(\theta,\sigma(m))
\end{equation*}
for all $\lambda'\in\Lambda$ such that $\lambda'_M=\lambda_M$. This is the definition of a credible profile in \cref{definition: credible}.

In addition, the Receiver should have no profitable deviation at every information set $(\lambda_M,m)$ on the equilibrium path. So for every $\rho'\in \Xi$ and every $m\in M$ such that $\lambda_M(m)>0$, we have
\begin{equation*}
\sum_\theta  \lambda(\theta,m)u_R(\theta,\rho(\lambda_M, m))\geq \sum_\theta \lambda(\theta,m)u_R(\theta,\rho'(\lambda_M, m)).
\end{equation*}
Note that $\sigma(m) = \rho(\lambda_M,m)$ for all $m\in M$, so by summing over all $m\in M$, we have
\begin{equation*}
\sum_\theta \sum_m \lambda(\theta,m)u_R(\theta,\sigma(m))\geq \sum_\theta \sum_m \lambda(\theta,m)u_R(\theta,\sigma'(m))
\end{equation*}
for all $\sigma'\in\Sigma$. This is the definition of a R-IC profile in \cref{definition: RIC}. Therefore, the profile $(\lambda,\sigma)$ is both credible and R-IC.

Moreover, notice that after the Sender chooses the uninformative test $\lambda^{\circ} = \mu_0 \times \delta_{m^{\circ}}$  for some $m^{\circ}$, the information set $(\delta_{m^{\circ}} ,m^{\circ})$ forms the initial node of a proper subgame. 
Subgame-perfection at this subgame then requires the Receiver to choose an action $a\in A_0$. Since the Sender always has the option to choose $\lambda^{\circ}$, her equilibrium payoff cannot be less than $\min_{a\in A_0} \sum_\theta\mu_{0}(\theta)u_S(\theta,a)$, so $\sum_\theta \sum_m \lambda(\theta,m)u_S(\theta,\sigma(m))\geq \min_{a\in A_0} \sum_\theta\mu_{0}(\theta)u_S(\theta,a)$.

The ``if" direction: Suppose that the profile $(\lambda,\sigma)$ is credible and suppose that the Sender's payoff from this profile is greater than the lowest possible no-information payoff. Consider the strategy profile $(\lambda,\rho)$ where $\rho(\lambda_M,m)=\sigma(m)$ for all $m\in M$, but for every $\lambda_M'\neq \lambda_M$, $\rho(\lambda_M',m)\in\arg\min_{a\in A_0} \sum_\theta\mu_{0}(\theta)u_S(\theta,a)$ for all $m\in M$. That is, the Receiver chooses the worst (w.r.t. the Sender's payoff) best response to prior belief after observing an off-path marginal distribution on $M$. We claim that $(\lambda,\rho)$ is an SPNE.

We first show that $\rho$ best responds to $\lambda$ in every subgame. We start with the extensive-form game itself. Note that from R-IC, the Receiver best responds to his on-path information sets $(\lambda_M,m)$; that is,
\begin{equation*}
\sum_\theta  \lambda(\theta,m)u_R(\theta,\rho(\lambda_M, m))\geq \sum_\theta \lambda(\theta,m)u_R(\theta,\rho'(\lambda_M, m)) \text{ for all } \rho'\in \Xi
\end{equation*}
at every $\lambda_M,m$ such that $\lambda_M( m)>0$. So $\rho$ best responds to $\lambda$ in the extensive-form game itself.

Next we consider proper subgames of the extensive-form game. Note that among all information sets, the only ones that form the initial node of a proper subgame are those in the form of $(\delta_{m^{\circ}}, m^{\circ})$, which is induced by the Sender choosing an uninformative test  $\lambda^{\circ} = \mu_0 \times \delta_{m^{\circ}}$  for some $m^{\circ}\in M$. By our construction, the Receiver chooses  the worst (w.r.t. the Sender's payoff) best response to prior belief, so the Receiver's strategy is a best response in these proper subgames.
Lastly, for any off-path information set $(\lambda_M',m)$ that does not define a subgame, SPE has no requirement on the Receiver's strategy. So $\rho$ best responds to $\lambda$ in every subgame.

We now turn to the Sender's strategy $\lambda$ and show that it best responds to $\rho$. From the credibility of $(\lambda,\sigma)$, the Sender has no incentive to deviate to any $\lambda'$ with $\lambda_M'=\lambda_M$. For any other deviation, her payoff is
\(
	\min_{a\in A_0} \sum_\theta\mu_{0}(\theta)u_S(\theta,a_0)
\),
which is her lowest no-information payoff. Since $\sum_\theta \sum_m \lambda(\theta,m)u_S(\theta,\sigma(m))\geq \min_{a\in A_0} \sum_\theta\mu_{0}(\theta)u_S(\theta,a)$, such deviations are not profitable. Therefore $\lambda$ best responds to $\rho$, so $(\lambda,\rho)$ is a pure-strategy SPNE and $(\lambda,\sigma)$ is the corresponding pure-strategy SPNE outcome.
\end{proof}

\section{Credible Persuasion in Games} \label{section: games}

In this section we generalize the framework in \cref{section: modelsetup} to a setting with multiple Receivers, where the Sender can also take actions after information is disclosed. We also allow the state space and action space to be infinite.

Consider an environment with a single Sender (she) and $r$ Receivers (each of whom is a he). The Sender has action set $A_S$ while each Receiver $ i\in \{1, \ldots, r \}$ has action set $A_i$. Let $A = A_S \times A_1 \times \ldots \times A_r $ denote the set of action profiles. 
Each player has payoff function $u_{i}: \Theta\times A  \rightarrow \mathbb{R}$, $i=S,1,\ldots,r$, respectively. The state space $\Theta$ and action spaces $A_i$ are Polish spaces endowed with their respective Borel sigma-algebras. Players hold full-support common prior $\mu_{0}\in\Delta(\Theta)$. We refer to $G=\left(\Theta, \mu_{0}, A_S, u_S, \{A_i \}_{i=1}^r, \{u_{i}\}_{i=1}^r \right)$ as the base game.

Let $M$ be a Polish space that contains $A$. The Sender chooses a test $\lambda\in \Delta(\Theta\times M)$ where $\lambda_{\Theta}=\mu_{0}$: note that this formulation implies that the test generates \textit{public} messages observed by all Receivers. Together the test and the base game constitute a Bayesian game $\mathcal{G} = \langle G,\lambda \rangle$, where:\footnote{The test $\lambda$ can be viewed as ``additional information'' observed by both the Sender and the Receivers, on top of the base information structure where the Sender observes the state and the Receivers do not observe any signal.}
\begin{enumerate}[itemsep=0em, topsep=0.3em]
    \item At the beginning  of the game a state-message pair $(\theta,m)$ is drawn from the test $\lambda$;
    \item The Sender observes $(\theta,m)$ while the Receivers observe only $m$; and
    \item All players choose an action simultaneously.
\end{enumerate}
A strategy profile $\sigma: \Theta \times M \rightarrow A$ in $\mathcal{G}$ consists of a Sender's strategy $\sigma_S: \Theta \times M \rightarrow A_S$ and Receivers' strategies $\sigma_{i}: M \rightarrow A_{i}$, $i=1,\ldots, r$. For each profile of Sender's test and players' strategies $(\lambda, \sigma)$, players' expected payoffs are given by  

\begin{equation*}
U_i(\lambda,\sigma)= \int_{\Theta \times M} u_i(\theta,\sigma(\theta, m))\,d\lambda(\theta,m) \;\;\; \text{ for }  \; i=S,
1,\ldots, r.
\end{equation*}

We now generalize the notion of credibility and incentive compatibility in \cref{section: model} to the current setting. For each $\lambda$, let
$D(\lambda)\equiv\big \{\lambda' \in \Delta(\Theta \times M): \lambda'_{\Theta}=\mu_{0},  \lambda'_M = \lambda_{M}  \big \} $
denote the set of tests that induce the same distribution of messages as $\lambda$. \cref{definition: credibe GAME} is analogous to \cref{definition: credible}, which requires that given the players' strategy profile, no deviation in $D(\lambda)$ can be profitable for the Sender.
\begin{definition} \label{definition: credibe GAME} 
A profile $(\lambda,\sigma)$ is \textbf{credible} if 
\begin{equation}\label{equation: credible GAME}
	\lambda \in \argmax_{\lambda'\in  D(\lambda)} \int u_S(\theta,\sigma(\theta, m))\,d\lambda'(\theta,m).
\end{equation}
\end{definition}

In addition, \cref{definition: IC GAME} generalizes \cref{definition: RIC}, and requires players' strategies to form a Bayesian Nash equilibrium of the game $\langle G,\lambda \rangle$.
\begin{definition} \label{definition: IC GAME}
A profile $(\lambda,\sigma)$ is \textbf{incentive compatible} (IC) if $\sigma$ is a Bayesian Nash equilibrium in $\mathcal{G} = \langle G,\lambda \rangle$. That is,
\begin{equation} \label{equation: IC GAME}
    \sigma_S \in \argmax_{\sigma_S':\Theta \times M \rightarrow A_S } U_S(\lambda, \sigma_S', \sigma_{-S} ) \quad \text{ and } \quad     \sigma_i \in \argmax_{\sigma_i': M \rightarrow A_i } U_i(\lambda, \sigma_i', \sigma_{-i} )\text{ for } i=1,\ldots,r.
\end{equation}
\end{definition}
Note that in  \cref{definition: credibe GAME}, when the Sender deviates to a different test, say $\lambda'$, we use the original strategy profile $\sigma(\theta,m)$ to predict players' actions in the ensuing Bayesian game $\langle G,\lambda' \rangle$. One might worry that the Sender may simultaneously change not only her test but also her strategy $\sigma_S(\theta,m)$ in $\langle G,\lambda' \rangle$. This, however, is unnecessary since the Sender's optimal strategy in $\langle G,\lambda' \rangle$ will remain unchanged: the Sender knows $\theta$ perfectly, her best response in $\langle G,\lambda' \rangle$ depends only on $\theta$ and the Receivers' actions (and not on her own test).

\end{document}